\documentclass[journal]{IEEEtran}
\usepackage{amsfonts}
\usepackage{amssymb}
\usepackage{amsthm}
\usepackage{amsmath,amsfonts,amssymb}
\usepackage[dvips]{graphicx}
\usepackage{verbatim}
\usepackage{setspace}
\usepackage{bm}
\usepackage{extarrows}
\usepackage{algorithmic} 
\usepackage[ruled,vlined]{algorithm2e}
\usepackage{cite}

\usepackage{changepage}
\usepackage{pdfpages}
\usepackage{color}
\usepackage{bbding}
\newtheorem{theorem}{Theorem}
\newtheorem{lemma}{Lemma}

\newcommand{\figwidth}{8}
\IEEEoverridecommandlockouts
\usepackage{etoolbox}
\let\mybibitem\bibitem
\renewcommand{\bibitem}[1]{%
	\ifstrequal{#1}{nature}
	{\color{blue}\mybibitem{#1}}
	{\color{black}\mybibitem{#1}}%
}

\begin{document}
%
\title{Joint 3-D Positioning and Power Allocation for UAV Relay Aided by Geographic Information} 

%
%
%
\author{Pengfei Yi,~\IEEEmembership{Graduate Student Member,~IEEE,}
	Liang Zhu,
	Lipeng Zhu,~\IEEEmembership{Member,~IEEE,}
    Zhenyu Xiao,~\IEEEmembership{Senior Member,~IEEE,}
    Zhu Han,~\IEEEmembership{Fellow,~IEEE,}
	and Xiang-Gen Xia,~\IEEEmembership{Fellow,~IEEE}

\thanks{This work was supported in part by the National Key Research and Development Program under grant number 2020YFB1806800, the National Natural Science Foundation of China (NSFC) under grant numbers 62171010 and 61827901, the Beijing Natural Science Foundation under grant number L212003, and NSF CNS-2107216 and CNS-2128368. The corresponding author is Dr. Zhenyu Xiao with Email xiaozy@buaa.edu.cn.}
\thanks{Pengfei Yi, Zhenyu Xiao are with the School of Electronic and Information Engineering, Beihang University, Beijing 100191, China. (yipengfei@buaa.edu.cn, zhulipeng@buaa.edu.cn, xiaozy@buaa.edu.cn)}
\thanks{Liang Zhu is with the IoT Division of CTTL-Terminals Laboratory,
	China Academy of Information and Communication Technology, Beijing 100191, China. (zhuliang@caict.ac.cn)}
\thanks{Lipeng Zhu is with the Department of Electrical and Computer Engineering, National University of Singapore, Singapore 117583 (zhulp@nus.edu.sg)}
\thanks{Zhu Han is with the Department of Electrical and Computer Engineering at the University of Houston, Houston, TX 77004 USA, and also with the Department of Computer Science and Engineering, Kyung Hee University, Seoul, South Korea, 446-701. (zhan2@uh.edu)}
\thanks{Xiang-Gen Xia is with the Department of Electrical and Computer Engineering, University of Delaware, Newark, DE 19716, USA. (xianggen@udel.edu)}
}

%
%

\maketitle

\begin{abstract}
In this paper, we study to employ geographic information to address the blockage problem of air-to-ground links between UAV and terrestrial nodes. In particular, a UAV relay is deployed to establish communication links from a ground base station to multiple ground users. To improve communication capacity, we first model the blockage effect caused by buildings according to the three-dimensional (3-D) geographic information. Then, an optimization problem is formulated to maximize the minimum capacity among users by jointly optimizing the 3-D position and power allocation of the UAV relay, under the constraints of link capacity, maximum transmit power, and blockage. To solve this complex non-convex problem, a two-loop optimization framework is developed based on Lagrangian relaxation. The outer-loop aims to obtain proper Lagrangian multipliers to ensure the solution of the Lagrangian problem converge to the tightest upper bound on the original problem. The inner-loop solves the Lagrangian problem by applying the block coordinate descent (BCD) and successive convex approximation (SCA) techniques, where UAV 3-D positioning and power allocation are alternately optimized in each iteration. Simulation results confirm that the proposed solution significantly outperforms three benchmark schemes and achieves a performance close to the upper bound on the UAV relay system.
\end{abstract}

\begin{IEEEkeywords}
UAV, relay communications, geographic information,  3-D positioning, power allocation.
\end{IEEEkeywords}

%
\IEEEpeerreviewmaketitle

\section{Introduction}
\IEEEPARstart{U}{nmanned} aerial vehicles (UAVs) have attracted increasing attention for enhancing the performance of wireless communication networks~\cite{gupta2016survey, zeng2019access, mozaffari2019atutor, fotouhi2019survey, xiao2020uavcom, song2021aerial}. Compared to conventional cellular systems, UAV-assisted  communications do not rely on fixed terrestrial infrastructures, and can be flexibly deployed over a target area in an on-demand and cost-effective manner. For instance, UAVs may serve as aerial base stations (BSs) or relays to support ground user equipments (UEs), as well as extend communication coverage in hot-spot regions or disaster areas. 

Benefiting from their three-dimensional (3-D) mobility, UAVs may adjust their positions or trajectories according to traffic demands to provide satisfactory communication service. As a result, a large number of works have been devoted to UAV-assisted communication systems~\cite{wu2018jointt, xu2020multiu, zhu2020millim, cai2018dualua, li2020datara, zhang2019securi,wang2018spectrum, liu2021resour, xiao2020unmann, zhang20213ddepl,zeng2021trajec}. 
For instance, in~\cite{wu2018jointt}, two-dimensional (2-D) trajectory, multiuser scheduling and association, and transmit power were jointly designed for a downlink multi-UAV enabled communication system, aiming to maximize the minimum throughout among all ground UEs.  
In~\cite{xu2020multiu},  2-D UAV trajectory and transmit beamforming vector were jointly optimized to minimize the total power consumption for multiuser downlink multiple-input single-output (MISO) UAV communication systems. UAV jittering, UE location uncertainty, wind speed uncertainty, and no-fly zones were taken into account to provide reliable communication services. 
In~\cite{zhu2020millim}, UAV positioning, analog beamforming, and power control were jointly optimized for a millimeter-wave full-duplex UAV relay communication system.  
A secure communication problem was investigated for a UAV-enabled communication system in~\cite{cai2018dualua}, where UAV trajectory and UE scheduling were jointly designed for maximizing the minimum secrecy rate. 
In~\cite{wang2018spectrum}, UAV trajectory and transmit power were jointly designed to achieve efficient spectrum sharing for a full-duplex UAV relaying system with underlaid device-to-device (D2D) communications.
A UAV-enabled energy-efficient Internet of Things (IoT) communication system was considered in~\cite{liu2021resour}, where 3-D placement and resource allocation of multiple UAV-mounted BSs were optimized to minimize the transmit power of IoT devices. 
In~\cite{zhang20213ddepl}, 3-D positions of UAVs, UE clustering, and frequency band allocation were optimized to improve the coverage rate as well as to minimize the number of required UAVs for a multiple UAV-BSs assisted communication system.

To make the problem more tractable, the aforementioned works mostly employ simplified or statistical channel models for UAV-assisted communication systems. For example, pure line-of-sight (LoS) channels are assumed in~\cite{wu2018jointt, xu2020multiu, cai2018dualua, zhang2019securi, li2020datara,wang2018spectrum}. Probabilistic LoS channels are adopted in~\cite{xiao2020unmann, liu2021resour, zeng2021trajec, zhang20213ddepl}, where LoS and non-LoS (NLoS) conditions are probabilistically related to the elevation angle of the link between the UAV and ground UE. 
However, the simplified and statistical channel models are only suitable for average performance analysis, while corresponding solutions may fail to provide satisfactory performance for practical UEs in specific environments. In practice, the terrain conditions,  such as  buildings and other obstacles, may block  air-to-ground signals and greatly deteriorate the propagation environment, especially for urban areas with dense buildings~\cite{bai2014analys}. 
In such a case, the design and optimization for UAV-assisted communication systems based on pure LoS channels or statistical channels may not be able to guarantee the performance of ground UEs.

\textcolor{black}{To overcome the drawback for using pure LoS channel models or probabilistic LoS channel models, two kinds of information are useful to capture practical propagation conditions, namely \emph{radio map} and \emph{geographic information} (or called a 3-D city map). Radio map, which is constructed based on large numbers of real-life channel measurements, can precisely describe the average signal strength for all combinations of UAV-UE/BS locations~\cite{chen2017learni}. Based on radio map, UAV positioning, UE association, and backhaul capacity allocation were jointly optimized to achieve higher capacity in a UAV-aided relay network in~\cite{liu2019optimi}. 
In~\cite{hu2020lowcom},	a joint trajectory design, UE scheduling, and bandwidth allocation optimization were investigated to ensure the fairness among UEs in a UAV-aided relay network. A trajectory design problem for cellular-connected UAV  was considered in~\cite{zhang2021radiom}, where radio map was utilized to evaluate the link quality between UAV and BS during the flight. Although radio map is theoretically appealing to provide the most precise channel quality, its practical application in UAV positioning is not easy. On one hand, it is difficult to obtain sufficiently accurate assessments of channel quality for all possible combinations of UAV-UE/BS 3-D  positions. One the other hand, even if a radio map is available, it may lead to high memory and computational cost related to data processing. In addition to radio map, geographic information is also helpful for efficient UAV positioning. In fact, with the geographic information available, the LoS/NLoS condition can be inferred directly by evaluating the BS-UAV/UAV-UE link blocked by ground buildings, instead of being modeled as a random event~\cite{zeng2019access}. Note that compared to radio map, geographic information is easier to obtain in reality. For example, geographic information can be derived offline from digital maps~\cite{kim2018intera}, such as Google Maps, or constructed online using photogrammetry techniques~\cite{zhou2013comple}. Based on geographic information, a \emph{local} LoS probability model was studied in~\cite{esrafilian2019learni} by employing the map-compression method, and then the UAV trajectory was designed to increase the data throughput in a UAV-aided wireless network. In~\cite{zhao2020effici}, a geometric analysis method to detect the blockage was proposed, and then a greedy UE scheduling algorithm was performed to avoid the blockage and enhance the spectral efficiency of a multi-UAV communication system. In~\cite{esrafilian2021threed}, UAV trajectory was designed for outdoor UE localization based on received signal strength measurements, where geographic information was utilized to distinguish LoS/NLoS environments.
}

Different from the above mentioned works, the goal of this paper is to model the blockage effect as tractable constraints, and jointly design with UAV 3-D positioning and resource allocation in an optimization manner, to guarantee practical communication performance in an arbitrary local area. Specifically, we consider a downlink multiuser communication system, where a UAV acts as an aerial relay to receive signals from the BS and forward them to ground UEs. Assisted by geographic information, we show that the blocked region for a BS-UAV/UAV-UE link caused by a ground building can be modeled as a polyhedron. Therefore, LoS and NLoS propagation conditions can be distinguished by evaluating whether the UAV is positioned in the blocked region.
Considering the blockage effect, an optimization problem is formulated to maximize the minimum communication capacity among ground UEs, and the corresponding solution is developed. The main contributions of this paper are summarized as follows.

\begin{enumerate}
  \item We propose to deploy a UAV relay to improve the performance of a multiuser communication system. Assisted by geographic information, the blockage effect caused by buildings is derived to guarantee practical performance. To ensure fairness, we formulate an optimization problem to maximize the minimum communication capacity among UEs by jointly designing the positioning and the power allocation at the UAV, where the backhaul link from the ground BS to the UAV relay is also taken into consideration.
  
  \item {\color{black}To solve the formulated non-convex problem which is very difficult to obtain the globally optimal solution, a two-loop optimization framework is developed based on Lagrangian relaxation to obtain a sub-optimal solution.} Specifically, the proper Lagrangian multipliers are adaptively updated in the outer-loop to ensure the Lagrangian problem converge to the tightest upper bound on the original problem. The inner-loop partitions the Lagrangian problem into a power allocation sub-problem and a UAV positioning sub-problem, and optimizes them alternately by applying the block coordinate descent (BCD) technique. The optimal power allocation is obtained in closed form, while the UAV positioning sub-problem is solved approximately by utilizing the successive convex approximation (SCA) technique.
  
  \item Simulation results show that the proposed geographic information-aided joint positioning and power allocation scheme can closely approach a performance upper bound on the considered UAV relay system and outperform {\color{black}{three}} benchmark schemes.
\end{enumerate}



The rest of this paper is organized as follows. In Section~\ref{sec_systemModel}, we introduce the system model, and formulate the proposed joint positioning and power allocation problem. In Section~\ref{sec_problem}, we transform the original problem into a more tractable Lagrangian problem. In Section~\ref{sec_solution_relaxed}, we propose a solution for the Lagrangian problem. The overall solution, convergence analysis, and computational complexity are discussed in Section~\ref{sec_overall}.  Section~\ref{sec_simulation} presents the simulation results. Finally, the paper is concluded in Section~\ref{sec_conclusion}.

\textit{Notation}: $a$, $\mathbf{a}$, $\mathbf{A}$, and $\mathcal{A}$ denote a scalar, a vector, a matrix, and a set, respectively. $\mathbb{R}^M$ denotes the space of $M$-dimensional real vector. $(\cdot)^{\rm{T}}$, $(\cdot)^{*}$, and $(\cdot)^{\rm{H}}$ denote transpose, conjugate, and conjugate transpose, respectively. $\|\mathbf{a}\|$ represents the Euclidean norm of vector $\mathbf{a}$. 
$\mathcal{A}_2 \setminus \mathcal{A}_1$ represents the elements of $\mathcal{A}_2$ that are not included in $\mathcal{A}_1$.
{\color{black}
$\mathcal{A}_1 \bigcup \mathcal{A}_2$ represents the union of $\mathcal{A}_1$ and $\mathcal{A}_2$. $\mathcal{A}_1 \subseteq \mathcal{A}_2$ represents that $\mathcal{A}_1$ is a subset of $\mathcal{A}_2$.}
$|\mathcal{A}|$ denotes the cardinality of set $\mathcal{A}$. $\emptyset$ denotes empty set. $\overrightarrow{AB}$ denotes the vector from point $A$ to point $B$. $\overrightarrow{AB} \cdot \overrightarrow{CD}$ and $\overrightarrow{AB} \times \overrightarrow{CD}$ denote the inner product and outer product between vector $\overrightarrow{AB}$ and vector $\overrightarrow{CD}$, respectively.

\section{System Model and Problem Formulation}\label{sec_systemModel}

As shown in Fig.~\ref{fig:system}, we consider a downlink communication network with a single BS and $K$ ground UEs indexed by $\mathcal{K}\triangleq \{1,...,K\}$. To improve the communication performance of the system, one UAV is employed as a decode-and-forward (DF) relay. That is to say, the BS transmits signals to the UAV relay, and then the UAV relay forwards the signals to the ground UEs in a half-duplex mode. 
\textcolor{black}{
We assume that the UAV allocates orthogonal frequency bands to the ground UEs. As a result, the interference among the UEs can be eliminated. In addition, it is assumed that the geographic information, i.e., the building structure, of the considered area is available to the system\footnote{\color{black}One publicly accessible way to extract building structure information is to use an open source geographic database \emph{OpenStreetMap}, where the contour and height of a building are given by its raw data tagged with \emph{geometry} and \emph{height}, respectively~\cite{kang2018buildi}.}.
} 

Without loss of generality, we employ a 3-D Cartesian coordinate system. The BS is located at $\mathbf{x}_{\mathrm{B}}^{} \in \mathbb{R}^3$. It is assumed that each UE $k$ is statically located at areas without buildings and its coordinates are denoted by $\mathbf{x}_k \in \mathbb{R}^3, k \in \mathcal{K}$. The coordinates of the UAV are given by $\mathbf{x}=(x_\mathrm{V},y_\mathrm{V},h_\mathrm{V}) \in \mathbb{R}^3$. $M$ buildings indexed by $\mathcal{M}\triangleq \{1,...,M\}$ are randomly distributed in the considered area. The BS-UAV and UAV-UE links may be blocked by these buildings. As a result, the position of the UAV relay should be carefully designed to avoid blockage. We suppose that the UAV hovers at an altitude above the tallest building in the considered region, such that no collision will occur. 

\begin{figure}[t]
\begin{center}
  \includegraphics[width=\figwidth cm]{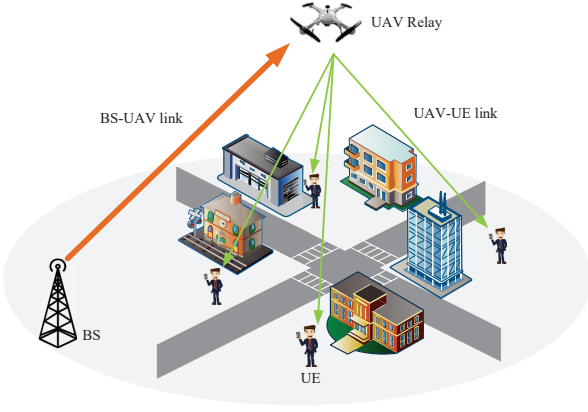}
  \caption{Illustration of the considered UAV relay communication system.}
  \label{fig:system}
\end{center}
\end{figure}

\subsection{Blockage Modeling}
To distinguish different propagation scenarios, the key is to identify whether the communication link is blocked by buildings. For $M$ buildings, $M$ and $MK$ blocked regions are generated with respect to the BS and $K$ UEs, respectively. We index these blocked regions as $\mathcal{I}\triangleq \{1,...,M+MK\}$. In this paper, we assume that the shape of all the buildings is cube\footnote{For other types of building, we can always find a cube envelope which covers the building, and thus the proposed solution is workable.}. We show that giving building and UE/BS locations, the blocked region can be formulated as a polyhedron. 

\subsubsection{Blocked Region for the UE}
As shown in Fig.~\ref{fig:surface_example}, along the UE's line of sight, only a part of the flank surfaces of each building is visible, while other flank surfaces and the top surface are invisible to the UE. In other words, the region behind the visible flank surfaces of the building is blocked. Therefore, in order to identify the blocked region for a UE with respect to a building, the key steps are a) identifying visible flank surfaces of the building, and b) determining the boundaries of the blocked region.

According to basic geometry, the inner product between the outward normal vector of a  surface and the LoS vector from a UE to any point on the surface can be utilized to judge whether the surface is visible. If the inner product is negative, the surface is \emph{visible}. Otherwise, the surface is \emph{invisible}~\cite{Hughes2013comput}. A toy example is shown in Fig.~\ref{fig:surface_example} (a), where surface $A_1B_1B_2A_2$ and surface $B_1C_1C_2B_2$ are visible to UE $S$.  The aerial view in Fig.~\ref{fig:surface_example} (b) gives a more visualized interpretation for the visible surfaces. 
\begin{figure}[t]
	\begin{center}
		\includegraphics[width=\figwidth cm]{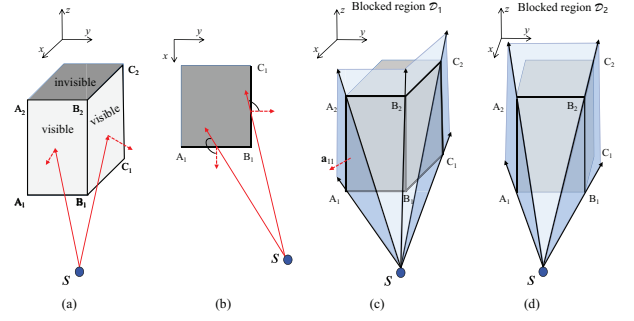}
		\caption{Illustration of the blockage caused by building.}
		\label{fig:surface_example}
	\end{center}
\end{figure}

As shown in Figs.~\ref{fig:surface_example} (c) and (d), the boundaries of the blocked region are composed by two vertical hyperplanes and two (or one) oblique hyperplanes\footnote{The number of oblique hyperplanes equals to the number of visible flank surfaces.}, where each hyperplane is determined by the position of the UE and two adjacent vertices on the flank surface of the building. 
Since a hyperplane determines two halfspaces, by applying analytic geometry theory, the blocked region can be represented by the intersection of a finite number of halfspaces, which is polyhedron~\cite{boyd2004convex}. The $i$-th blocked region is
\begin{equation}
	\mathcal{D}_i=\{ \mathbf{x} \in \mathbb{R}^3 | \mathbf{a}_{ij}^{\rm{T}} \mathbf{x} -b_{ij} \le 0, \forall j\in \mathcal{J}_i \},
\end{equation}
where $\mathcal{J}_i$ is the set of indexes of hyperplanes for the $i$-th blocked region. $|\mathcal{J}_i|$ equals to four or three, corresponding to situations in Figs.~\ref{fig:surface_example} (c) and (d). $\mathbf{a}_{ij}\in\mathbb{R}^3$ is the outward normal vector of hyperplane $j$, which can be obtained by the outer product of any two non-collinear vectors in the plane, pointing to the outside of the blocked region. Constant $b_{ij} \in \mathbb{R}$ determines the offset of hyperplane $j$ from the origin point, which can be obtained by the inner product between $\mathbf{a}_{ij}$ and any point in the hyperplane. For example, in Fig.~\ref{fig:surface_example} (c), the outer boundaries of the blocked region $\mathcal{D}_1$ are hyperplane $SA_1A_2$, $SA_2B_2$, $SB_2C_2$, and $SC_2C_1$. The outward normal vector $\mathbf{a}_{11}=\frac{\overrightarrow{SA_2} \times \overrightarrow{SA_1}}{\|\overrightarrow{SA_2} \times \overrightarrow{SA_1}\|}$, and the corresponding offset $b_{11}=\mathbf{a}_{11}^{\rm{T}} \cdot \overrightarrow{OS}$, where $O$ denotes the origin point.

In summary, given the location of the UE and the vertices of a building, the algorithm for obtaining the blocked region is shown in Algorithm~\ref{alg:blckage_formula}.
\begin{algorithm}[t] \small
	\label{alg:blckage_formula}
	\caption{Blocked region.}
	\begin{algorithmic}[1]
		\REQUIRE ~The location of a UE and the vertices of a building.
		\ENSURE ~The blocked region $\mathcal{D}_i$. \\
		\FOR {\textbf{each} flank surface of the building}
		\STATE Calculate the outward normal vector, and select a point on the surface to form an LoS vector with the UE.
		\IF  {the inner product between the two vectors is negative}
		\STATE Set the surface as \emph{visible}.
		\ELSE
		\STATE Set the surface as \emph{invisible}.
		\ENDIF
		\ENDFOR
		\STATE Record all the edges of the visible surfaces, except for the bottom edges.
		\FOR {\textbf{each} recorded edge}
		\STATE Construct hyperplane determined by the edge and the UE's position. Denote corresponding outward normal vector as $\mathbf{a}_{ij}$ and offset as $b_{ij}$, $j\in \mathcal{J}_i$.
		\ENDFOR		
		\RETURN The blocked region is $\mathcal{D}_i=\{ \mathbf{x} \in \mathbb{R}^3 | \mathbf{a}_{ij}^{\rm{T}} \mathbf{x} -b_{ij} \le 0, \forall j\in \mathcal{J}_i \}$.
	\end{algorithmic}
\end{algorithm}

\subsubsection{Blocked Region for the BS}
Depending on the relative height of the BS and the buildings, there are two possible situations for the blocked region. If the BS antenna is deployed higher than {a building, BS-UAV link is not blocked by the building for all possible UAV positions.} 
The blocked region is an empty set, i.e., $\mathcal{D}_i=\emptyset$. On the other hand, if the BS antenna is lower than {a building}, the blockage modeling method for the UE in Algorithm~\ref{alg:blckage_formula} can be directly used for the BS.

\subsection{Channel Model}
{\color{black}On one hand, accurate state of NLoS components for BS-UAV/UAV-UE links is difficult to predict in practice, unless the UAV is deployed at a given fixed position and channel estimation is performed. On the other hand, the NLoS environment leads to dramatic deterioration of communication performance. Therefore, it is preferred to avoiding all blocked regions  ($\mathcal{D}_i, i\in \mathcal{I}$) to guarantee the overall performance.} 
{\color{black}
Let $\mathcal{D}=\{ \mathbf{x}=(x_\mathrm{V},y_\mathrm{V},h_\mathrm{V}) | x_\mathrm{V}\in [0,x_\mathrm{D}], y_\mathrm{V}\in [0,y_\mathrm{D}], h_\mathrm{V}\geq h_\mathrm{min} \}$ denote the whole considered region, then the position of the UAV relay should be restricted by
\begin{align}
\mathbf{x} \in \mathcal{D} \setminus \mathcal{D}_i, \forall i \in \mathcal{I}. \label{c:block}	
\end{align}
Under constraint (\ref{c:block}), the LoS propagation environment is always guaranteed for the BS-UAV and UAV-UE links\footnote{We assume that all UEs are located at areas without buildings. Therefore, the feasible region for UAV positioning is nonempty as there always exists an unblocked region for sufficient high altitude of the UAV relay.}. 
In this way, LoS channel models are adopted for UAV positioning design, while NLoS channels are not necessary in the design but required in the simulations.} Since the UAV positioning should be optimized in a relatively large timescale, we mainly focus on the large-scale channel characteristics.
 The distance from the BS to the UAV relay is $\|\mathbf{x}-\mathbf{x}_\mathrm{B}\|$, and the distance from the UAV relay to UE $k$ is $\|\mathbf{x}-\mathbf{x}_k\|$. Then the channel gains of BS-UAV and UAV-UE links are respectively given by
{
\begin{equation}\label{eq_channel_BS}
	g_{\mathrm{B}}^{}=
	\beta_1 \|\mathbf{x}-\mathbf{x}_\mathrm{B}\|^{-\alpha_1}, 
\end{equation}
\begin{equation}\label{eq_channel_User}
	g_k=
	\beta_1 \|\mathbf{x}-\mathbf{x}_k\|^{-\alpha_1}, 
\end{equation}}
where \textcolor{black}{$\alpha_1$ is} the path loss exponent, and  \textcolor{black}{$\beta_1$ is} the channel gain at reference distance of $1$ m.


As a result, we  obtain the communication capacities of the BS-UAV link and the UAV-UE link as follows:
\begin{equation}
R_{\mathrm{B}}= W_{\mathrm{B}} \log_2\left(1+\frac{P_{\mathrm{B}} \beta_1}{N_0 W_{\mathrm{B}} \|\mathbf{x}-\mathbf{x}_{\mathrm{B}}\|^{\alpha_1}}\right),
\end{equation}	
\begin{equation}
R_k= W_\mathrm{U} \log_2\left(1+\frac{P_k \beta_1}{N_0 W_\mathrm{U} \|\mathbf{x}-\mathbf{x}_k\|^{\alpha_1}}\right), 
\end{equation}
where $W_\mathrm{B}$ and $W_\mathrm{U}$ denote the channel bandwidths of links from the BS to the UAV and from the UAV to UE $k$, respectively. $P_\mathrm{B}$ denotes the transmit power of BS, and $P_k$ is the transmit power of the UAV allocated to UE $k$. $N_0$ is the power spectral density of additive white Gaussian noise.

\subsection{Problem Formulation}
To ensure the fairness, we aim to maximize the minimum communication capacity among all the UEs by jointly optimizing the UAV positioning and the power allocation, subject to the blockage effect of buildings, and the backhaul constraint. 
The problem is formulated as 
\begin{align}
	\max\limits_{\mathbf{x}, P_{\mathrm{B}}, \{P_k\}}~~ & R      \label{eq_problem}\\
	\mbox{s.t.}~~ &(\ref{c:block}),   \notag \\
	&R \leq R_k, \forall k\in \mathcal{K},   \tag{\ref{eq_problem}a} \label{c:rate1} \\
&K R \leq R_{\mathrm{B}}^{},  \tag{\ref{eq_problem}b}\label{c:rate2}  \\
&P_{k} \geq 0, \forall k\in \mathcal{K},  \tag{\ref{eq_problem}c}\label{c:p1}  \\
&\sum_{k\in \mathcal{K}} P_{k} \leq P_{\mathrm{V}}^{\mathrm{tot}},  \tag{\ref{eq_problem}d}\label{c:p2} \\
&0 \leq P_{\mathrm{B}}^{} \leq P_{\mathrm{B}}^{\mathrm{tot}},  \tag{\ref{eq_problem}e}\label{c:p3} 
\end{align}
where $R$ denotes the minimum communication capacity among all the UEs. {\color{black}Constraint (\ref{c:block}) confines that the UAV has to be deployed to avoid all the blocked regions for the BS-UAV and UAV-UE links.}
Constraint (\ref{c:rate1}) indicates that the minimum communication capacity does not exceed the communication capacity of each UAV-UE link. Constraint (\ref{c:rate2}) ensures that the BS-UAV backhaul link is capable to support all UEs with the minimum communication capacity. Constraints (\ref{c:p1}), (\ref{c:p2}), and (\ref{c:p3}) indicate that the transmit powers are nonnegative and do not exceed a maximum value, where $P_{\mathrm{B}}^{\mathrm{tot}}$ and $P_{\mathrm{V}}^{\mathrm{tot}}$ are the maximum transmit powers of the BS and the UAV relay, respectively.

\section{Problem Transformation}\label{sec_problem}
Due to the non-convex and intractable constraints, the original problem~(\ref{eq_problem}) is challenging to solve. In this section, we first transform the blockage constraint into tractable mixed-integer linear constraints. Then, Lagrangian relaxation is introduced to relax the binary variables to continuous variables  thus leading to a Lagrangian problem. 

\subsection{Transformation of the Blockage Constraint}\label{subsec_blockage}

To make constraint~(\ref{c:block}) more tractable, we introduce auxiliary binary variable $l_{ij}$ and provide the following theorem.

\begin{theorem} \label{Theo_blockage}
For each $i\in \mathcal{I}, \mathbf{x} \in \mathcal{D} \setminus \mathcal{D}_i$ is equivalent to 
\begin{equation}\label{eq_theo_block} 
	\left \{
	\begin{aligned}
		& \mathbf{a}_{ij}^T \mathbf{x} -b_{ij} +Cl_{ij} > 0,  \forall j\in \mathcal{J}_i,\\
	& l_{ij}\in \{0,1\},  \forall j\in \mathcal{J}_i,  \\
	& \sum_{j \in \mathcal{J}_i} l_{ij}\leq |\mathcal{J}_i|-1,\\ 
	& \mathbf{x} \in \mathcal{D},
	\end{aligned}
	\right.
\end{equation}
where $C$ is a sufficiently large constant.
\end{theorem}

\begin{proof}
On one hand, we verify that any point $\mathbf{x} \in \mathcal{D} \setminus \mathcal{D}_i$ satisfies constraint~(\ref{eq_theo_block}). The definition of $\mathcal{D} \setminus \mathcal{D}_i$ is given by $\mathcal{D} \setminus \mathcal{D}_i=\{ \mathbf{x} \in \mathcal{D} | \mathbf{a}_{ij}^T \mathbf{x} -b_{ij}> 0, \exists j \in \mathcal{J}_i\}$. With involving auxiliary binary variable $l_{ij}\in \{0,1\},  \forall j\in \mathcal{J}_i$, we know that $\exists j \in \mathcal{J}_i$ such that $\mathbf{a}_{ij}^T \mathbf{x} -b_{ij}+Cl_{ij}> 0$ by setting $l_{ij}=0$. For other $j\in \mathcal{J}_i$,  $\mathbf{a}_{ij}^T \mathbf{x} -b_{ij}+Cl_{ij}> 0$ holds by setting $l_{ij}=1$, for a sufficiently large $C>0$. Thus, $\forall j\in \mathcal{J}_i$, $\mathbf{a}_{ij}^T \mathbf{x} -b_{ij}+Cl_{ij}> 0$ holds, and there exists at least one $l_{ij}=0$, which makes $\sum_{j \in \mathcal{J}_i} l_{ij}\leq |\mathcal{J}_i|-1$ hold.

On the other hand, $l_{ij}\in \{0,1\},  \forall j\in \mathcal{J}_i$ and $\sum_{j \in \mathcal{J}_i} l_{ij}\leq |\mathcal{J}_i|-1$ constrains that $\exists j \in \mathcal{J}_i$, $l_{ij}=0$. Together with $\mathbf{x} \in \mathcal{D}$ and $\mathbf{a}_{ij}^T \mathbf{x} -b_{ij} +Cl_{ij} > 0,  \forall j\in \mathcal{J}_i$, the position $\mathbf{x}$ is constrained in a region where $\mathbf{x} \in \mathcal{D}$ and $\exists j \in \mathcal{J}_i$, $\mathbf{a}_{ij}^T \mathbf{x} -b_{ij} +Cl_{ij}= \mathbf{a}_{ij}^T \mathbf{x} -b_{ij}+0 > 0$, which is the definition of $\mathcal{D} \setminus \mathcal{D}_i$.

This completes the proof.
\end{proof}
According to Theorem~\ref{Theo_blockage}, constraint~(\ref{c:block}) is equivalent to the following mixed integer linear constraints:
\begin{align}
& \mathbf{a}_{ij}^T \mathbf{x} -b_{ij} +Cl_{ij} > 0, \forall j\in \mathcal{J}_i, \forall i\in \mathcal{I}, \label{c:block1} \\
& l_{ij}\in \{0,1\},   \forall j\in \mathcal{J}_i, \forall i\in \mathcal{I}, \label{c:block2} \\
& \sum_{j \in \mathcal{J}_i} l_{ij} \leq |\mathcal{J}_i|-1, \forall i\in \mathcal{I},  \label{c:block3} \\
& \mathbf{x} \in \mathcal{D},  \label{c:block4}
\end{align}
where $C$ is a sufficiently large constant.

\subsection{Lagrangian Relaxation}
By replacing constraint~(\ref{c:block}) with constraints (\ref{c:block1})--(\ref{c:block4}), problem~(\ref{eq_problem}) is equivalent to 
\begin{align}
	\max\limits_{\mathbf{x}, P_\mathrm{B}, \{P_k\}, \{l_{ij}\}}~~ & R \label{eq_problem_equivalent} \\
	\mbox{s.t.}~~ & \text{(\ref{c:rate1}), (\ref{c:rate2}), (\ref{c:p1}), (\ref{c:p2}), (\ref{c:p3}), }   \notag  \\
	&\text{(\ref{c:block1}), (\ref{c:block2}), (\ref{c:block3}), (\ref{c:block4})}.   \notag 
\end{align}

Due to the non-convexity of constraints (\ref{c:rate1}), (\ref{c:rate2}) and binary constraint (\ref{c:block2}), problem~(\ref{eq_problem_equivalent}) is a non-convex and mixed-integer problem, which is difficult to obtain a globally optimal solution. Therefore, we propose to utilize Lagrangian relaxation~\cite{Fisher1981The}. First, binary constraint (\ref{c:block2}) can be equivalently replaced by 
\begin{align}
	&0 \leq l_{ij} \leq 1,  \forall j\in \mathcal{J}_i, \forall i\in \mathcal{I}, \label{c:lij1} \\
	&\sum_{j \in \mathcal{J}_i} l_{ij}(1-l_{ij})\leq 0,  \forall i\in \mathcal{I}. \label{c:lij2}
\end{align}

In this way, $l_{ij}$ becomes a continuous optimization variable between 0 and 1. \textcolor{black}{By replacing constraint~(\ref{c:block2}) with constraints~(\ref{c:lij1}) and (\ref{c:lij2}), and then dualizing and penalizing constraint (\ref{c:lij2}) into the objective function with Lagrangian multipliers $\lambda_i$ for $i\in \mathcal{I}$,} 
%
we obtain a Lagrangian problem 
\begin{align}
	\max\limits_{\mathbf{x}, P_\mathrm{B}, \{P_k\}, \{l_{ij}\}}~~ & R-\sum_{i \in \mathcal{I}} \lambda_i \sum_{j \in \mathcal{J}_i} l_{ij}(1-l_{ij})   \label{eq_problem_relaxed} \\
	\mbox{s.t.}~~ &\lambda_i \geq 0, \forall i \in \mathcal{I} \tag{\ref{eq_problem_relaxed}a}\label{c:lambda},  \\
	&\text{(\ref{c:rate1})}, \text{(\ref{c:rate2}), (\ref{c:p1}), (\ref{c:p2}), (\ref{c:p3}),}   \notag\\
	&\text{(\ref{c:block1}), (\ref{c:block3}), (\ref{c:block4}), (\ref{c:lij1})}.   \notag
\end{align}
\textcolor{black}{
Comparing Lagrangian problem~(\ref{eq_problem_relaxed}) with original problem~(\ref{eq_problem_equivalent}), we have the following Lemma.
\begin{lemma} \label{Lemma_LR_upperbound}
For any given Lagrangian multipliers, the optimal value of Lagrangian problem~(\ref{eq_problem_relaxed}) yields an upper bound on the optimal value of original problem~(\ref{eq_problem_equivalent}). 
\end{lemma}
\begin{proof}See Appendix~\ref{app_LR_upperbound}.
\end{proof}
}

To decrease the duality gap between Lagrangian problem~(\ref{eq_problem_relaxed}) and original problem~(\ref{eq_problem_equivalent}) and obtain a good feasible solution to original problem~(\ref{eq_problem_equivalent}), we develop a two-loop optimization framework to adaptively determine the values of Lagrangian multipliers as well as the UAV positioning and power allocation.


\section{Solution of the Lagrangian problem}\label{sec_solution_relaxed}
In this section, we introduce the method to solve Lagrangian problem~(\ref{eq_problem_relaxed}) as well as adjusting the Lagrangian multipliers to deduce a good feasible solution to the original problem~(\ref{eq_problem_equivalent}). A two-loop optimization framework is developed, where the inner-loop optimizes the Lagrangian problem~(\ref{eq_problem_relaxed}) for given Lagrangian multipliers, and the outer-loop updates the multipliers to decrease the duality gap. The flowchart of the solution is shown in Fig.~\ref{fig:flowchart}. 
Specifically, due to the non-concave objective function and non-convex constraints in (\ref{c:rate1}) and (\ref{c:rate2}), Lagrangian problem~(\ref{eq_problem_relaxed}) is not a convex problem and is challenging to obtain a globally optimal solution. Therefore, we develop an iterative algorithm by applying the BCD technique~\cite{xing2021matrix,xing2021matrix2}. For given UAV position $\mathbf{x}$ and the auxiliary variables $\{l_{ij}\}$, we solve the power allocation sub-problem and obtain a closed-form solution. For given transmit powers $\{P_k\}$ and $P_{\mathrm{B}}$, we solve the UAV positioning sub-problem approximately by utilizing the SCA technique~\cite{Dinh2010local}. These two sub-problems are alternately optimized until a suboptimal solution is obtained, which are discussed in Sections~\ref{subsec_PA} and \ref{subsec_Pos}. The method to update Lagrangian multipliers is discussed in Section~\ref{subsec_multiplier}.

\begin{figure}[t]
	\begin{center}
		\includegraphics[width=0.6\linewidth]{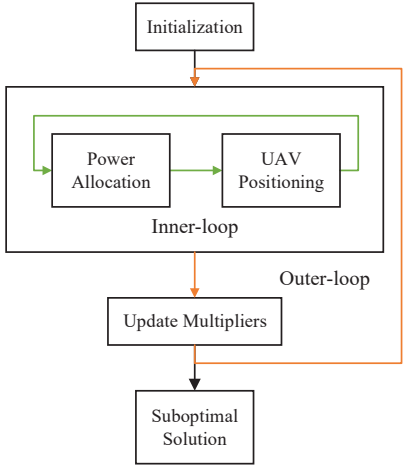}
		\caption{The flowchart of two-loop solution for Lagrangian problem~(\ref{eq_problem_relaxed}).}
		\label{fig:flowchart}
	\end{center}
\end{figure}

\subsection{Power Allocation}\label{subsec_PA}
For the $T$-th outer-loop, given Lagrangian multipliers $\{\lambda_i^{(T)}\}$, we develop an iterative algorithm to solve Lagrangian problem~(\ref{eq_problem_relaxed}). For the $t$-th iteration of the inner-loop, given UAV's position $\mathbf{x}^{(t)}$ and auxiliary variables $\{l_{ij}^{(t)}\}$, problem~(\ref{eq_problem_relaxed}) is transformed to the following power allocation problem 
\begin{align}
	\max\limits_{P_\mathrm{B}, \{P_k\}}~~ & R - \sum_{i \in \mathcal{I}} \lambda_i^{(T)} \sum_{j \in \mathcal{J}_i} l_{ij}^{(t)}(1-l_{ij}^{(t)})  \label{eq_problem_PowerControl} \\
	\mbox{s.t.}~~ &R \leq W_\mathrm{U} \log_2\left(1+\eta_k^{(t)} P_k \right), \forall k\in \mathcal{K}, \tag{\ref{eq_problem_PowerControl}a} \\
	&K R \leq W_{\mathrm{B}}^{} \log_2\left(1+\eta_\mathrm{B}^{(t)} P_{\mathrm{B}}^{} \right),  \tag{\ref{eq_problem_PowerControl}b} \\
	&\text{(\ref{c:p1}), (\ref{c:p2}), (\ref{c:p3})}, \notag
\end{align}
where 
$\eta_k^{(t)}=\beta_1/ (N_0 W_\mathrm{U} \|\mathbf{x}^{(t)}-\mathbf{x}_k\|^{\alpha_1})$ and $\eta_{\mathrm{B}}^{(t)}=\beta_1/(N_0 W_{\mathrm{B}} \|\mathbf{x}^{(t)}-\mathbf{x}_{\mathrm{B}}\|^{\alpha_1})$ 
are intermediate parameters for notational simplicity. The second term in the objective function is a constant which does not impact the optimality. 

Problem~(\ref{eq_problem_PowerControl}) is a convex problem with respect to $P_\mathrm{B}$ and $ \{P_k\}$~\cite{zhu2019millim, diamond2016cvxpy}. To maximize the minimum communication capacity with minimal transmit powers, we provide the following theorem.

\begin{theorem} \label{Theo_power}
	For given position $\mathbf{x}^{(t)}$ and auxiliary variables $\{l_{ij}^{(t)}\}$, the optimal power allocation for the BS and UAV relay is given as follows:
\begin{align}
	&~\rm{If}~\left(1+\eta_\mathrm{B}^{(t)} P_{\mathrm{B}}^{\mathrm{tot}}\right)^{\frac{W_{\mathrm{B}}^{}}{K}}< \left(1+\eta_\mathrm{V}^{(t)} P_{\mathrm{V}}^{\mathrm{tot}} \right)^{W_{\mathrm{U}}^{}}, \notag  \\ 
	&\left\{
	\begin{aligned}
		&P_{\mathrm{B}}^{(t+1)}=P_{\mathrm{B}}^{\mathrm{tot}},\\
		&P_{k}^{(t+1)}=\left( \left(1+\eta_\mathrm{B}^{(t)} P_{\mathrm{B}}^{\mathrm{tot}}\right)^{\frac{W_{\mathrm{B}}^{}}{KW_{\mathrm{U}}}}-1 \right)/\eta_k^{(t)};
	\end{aligned}
	\right. \notag
\end{align}
\begin{align}
	&~\rm{otherwise},  \notag \\
	&\left\{
	\begin{aligned}
		&P_{\mathrm{B}}^{(t+1)} = \left( \left(1+\eta_\mathrm{V}^{(t)} P_{\mathrm{V}}^{\mathrm{tot}} \right)^{\frac{KW_{\mathrm{U}}^{}}{W_{\mathrm{B}}^{}}} - 1 \right) / \eta_\mathrm{B}^{(t)},\\
		&P_k^{(t+1)}=\eta_\mathrm{V}^{(t)}P_{\mathrm{V}}^{\mathrm{tot}}/\eta_k^{(t)} , \label{eq_opt_power}
	\end{aligned}
	\right.
\end{align}
with $\eta_\mathrm{V}^{(t)}=1/\sum \limits_{k \in \mathcal{K}}\frac{1}{\eta_k^{(t)}}$.
\end{theorem}

\begin{proof}
See Appendix~\ref{app_power}.
\end{proof}

Hereto, we have obtained the optimal solution for the power allocation sub-problem. 
%

\subsection{UAV Positioning}\label{subsec_Pos}

Given the power allocation $\{P_k^{(t+1)}\}$ and $P_{\mathrm{B}}^{(t+1)}$, problem~(\ref{eq_problem_relaxed}) is transformed to the following positioning problem
\begin{align}
	\max\limits_{\mathbf{x}, \{l_{ij}\}}~~ & R-\sum_{i \in \mathcal{I}} \lambda_i^{(T)} \sum_{j \in \mathcal{J}_i} l_{ij}(1-l_{ij})  \label{eq_problem_positioning} \\
	\mbox{s.t.}~~ & R \leq W_\mathrm{U} \log_2\left(1+\zeta_k^{(t)}/ \|\mathbf{x}-\mathbf{x}_k\|^{\alpha_1}\right), \forall k\in \mathcal{K},  \tag{\ref{eq_problem_positioning}a} \label{c:pos_rate1} \\
	&  K R \leq W_{\mathrm{B}}^{} \log_2\left(1+\zeta_{\mathrm{B}}^{(t)}/ \|\mathbf{x}-\mathbf{x}_{\mathrm{B}}^{}\|^{\alpha_1}\right),   \tag{\ref{eq_problem_positioning}b} \label{c:pos_rate2} \\
	& \text{(\ref{c:block1}), (\ref{c:block3}), (\ref{c:block4}), (\ref{c:lij1}),}   \notag 
\end{align}
where $\zeta_k^{(t)} = \frac{P_k^{(t+1)}\beta_1}{N_0 W_\mathrm{U}}$ and $\zeta_{\mathrm{B}}^{(t)} = \frac{P_{\mathrm{B}}^{(t+1)}\beta_1}{N_0 W_{\mathrm{B}}}$ are intermediate parameters for notational simplicity. Problem~(\ref{eq_problem_positioning}) is not a concave/convex problem because of the non-concave objective function and the non-convex constraints in (\ref{c:pos_rate1}) and (\ref{c:pos_rate2}). In the following, we employ the SCA technique to obtain a sub-optimal solution. 

Considering the objective function, for a given local point $l_{ij}^{(t)}$ in the $t$-th iteration, we have the following lower bound on $(l_{ij})^2$:	
\begin{equation}\label{eq_l_lb}
(l_{ij})^2 \geq 2 l_{ij}^{(t)} l_{ij} - (l_{ij}^{(t)})^2.
\end{equation}

Considering constraints~(\ref{c:pos_rate1}) and  (\ref{c:pos_rate2}), note that function $r(z)=\log_2 \left(1+1/z^{\alpha_1} \right)$ is convex with respect to $z$ for $z>0$, thus the right-hand-side (RHS) of constraint~(\ref{c:pos_rate1}) and the RHS of constraint~(\ref{c:pos_rate2}) are convex with respect to $\|\mathbf{x}-\mathbf{x}_k\|$ and $\|\mathbf{x}-\mathbf{x}_{\mathrm{B}}^{}\|$, respectively. For a given local point $\mathbf{x}^{(t)}$ in the $t$-th iteration, the RHS of constraint~(\ref{c:pos_rate1}) is lower-bounded by its first-order Taylor expansion at $\|\mathbf{x}^{(t)}-\mathbf{x}_k\|$ as
\begin{equation}\label{eq_rate_lb}
	\begin{aligned}
		&W_\mathrm{U} \log_2\left(1+\zeta_k^{(t)}/ \|\mathbf{x}-\mathbf{x}_k\|^{\alpha_1}\right) \geq\\ &A_k^{(t)}-B_k^{(t)}(\|\mathbf{x}-\mathbf{x}_k\|-\|\mathbf{x}^{(t)}-\mathbf{x}_k\|)\triangleq R_k^\mathrm{lb},
	\end{aligned}
\end{equation}
with 
\begin{align}
	A_k^{(t)}&=W_\mathrm{U} \log_2\left(1+\zeta_k^{(t)}/ \|\mathbf{x}^{(t)}-\mathbf{x}_k\|^{\alpha_1}\right),\\
	B_k^{(t)}&=\frac{ W_\mathrm{U}\zeta_k^{(t)} \alpha_1  }{\|\mathbf{x}^{(t)}-\mathbf{x}_k\|(\|\mathbf{x}^{(t)}-\mathbf{x}_k\|^{\alpha_1}+\zeta_k^{(t)})\ln 2}.
\end{align}
The RHS of constraint~(\ref{c:pos_rate2}) is lower-bounded by its first-order Taylor expansion at $\|\mathbf{x}^{(t)}-\mathbf{x}_\mathrm{B}\|$ as
\begin{equation}\label{eq_rate_BS_lb}
	\begin{aligned}
		&W_{\mathrm{B}}^{} \log_2\left(1+\zeta_{\mathrm{B}}^{(t)}/ \|\mathbf{x}-\mathbf{x}_{\mathrm{B}}^{}\|^{\alpha_1}\right)\geq \\
		&A_{\mathrm{B}}^{(t)}-B_{\mathrm{B}}^{(t)}(\|\mathbf{x}-\mathbf{x}_{\mathrm{B}}\|-\|\mathbf{x}^{(t)}-\mathbf{x}_{\mathrm{B}}\|)\triangleq R_{\mathrm{B}}^\mathrm{lb},
	\end{aligned}
\end{equation}
with 
\begin{align}
	A_{\mathrm{B}}^{(t)}&=W_\mathrm{B} \log_2\left(1+\zeta_{\mathrm{B}}^{(t)}/ \|\mathbf{x}^{(t)}-\mathbf{x}_{\mathrm{B}}\|^{\alpha_1}\right),\\
	B_{\mathrm{B}}^{(t)}&=\frac{ W_\mathrm{B}\zeta_{\mathrm{B}}^{(t)} \alpha_1  }{\|\mathbf{x}^{(t)}-\mathbf{x}_{\mathrm{B}}\|(\|\mathbf{x}^{(t)}-\mathbf{x}_{\mathrm{B}}\|^{\alpha_1}+\zeta_{\mathrm{B}}^{(t)})\ln 2}.
\end{align}

As a result, given local points $\mathbf{x}^{(t)}$ and $\{l_{ij}^{(t)}\}$ in the $t$-th iteration, problem~(\ref{eq_problem_positioning}) is relaxed as
\begin{align}
\max\limits_{\mathbf{x}, \{l_{ij}\}}~~ & R-\sum_{i \in \mathcal{I}} \lambda_i^{(T)} \sum_{j \in \mathcal{J}_i} \left(l_{ij}- 2 l_{ij}^{(t)} l_{ij} + (l_{ij}^{(t)})^2 \right) \label{eq_problem_positioning_approx} \\
\mbox{s.t.}~~ &  R \leq R_k^\mathrm{lb}, \forall k\in \mathcal{K},  \tag{\ref{eq_problem_positioning_approx}a}\label{c:pos_rate11} \\
& K R \leq R_\mathrm{B}^\mathrm{lb},  \tag{\ref{eq_problem_positioning_approx}b}\label{c:pos_rate22}  \\
& \|\mathbf{x}-\mathbf{x}^{(t)}\| \leq \rho^{(t)},  \tag{\ref{eq_problem_positioning_approx}c}\label{trustregion}\\
&\text{(\ref{c:block1}), (\ref{c:block3}), (\ref{c:block4}), (\ref{c:lij1})}.   \notag 
\end{align}

It is easy to verify that problem~(\ref{eq_problem_positioning_approx}) is a convex problem and can be readily solved by standard convex program solvers, such as CVXPY~\cite{diamond2016cvxpy, xiao2018jointp}. The objective function of problem~(\ref{eq_problem_positioning_approx}) is a lower bound on the objective function of problem~(\ref{eq_problem_positioning}). Constraints~(\ref{c:pos_rate11}) and (\ref{c:pos_rate22}) have tighter forms compared to constraints~(\ref{c:pos_rate1}) and (\ref{c:pos_rate2}). Hance,
any feasible solution of problem~(\ref{eq_problem_positioning_approx}) is also feasible for problem~(\ref{eq_problem_positioning}). Constraint (\ref{trustregion}) is introduced to limit the update of $\mathbf{x}^{(t+1)}$ in a small neighborhood around the local point $\mathbf{x}^{(t)}$. $\rho^{(t)}$ is the radius of the spherical region and is gradually reduced during the iteration to ensure convergence. A feasible way to update $\rho^{(t)}$ is $\rho^{(t+1)}=\kappa_1 \rho^{(t)}$, where $\kappa_1<1$ is the stepsize. 

\subsection{Updating Lagrangian Multipliers}\label{subsec_multiplier}
To decrease the duality gap between Lagrangian problem~(\ref{eq_problem_relaxed}) and original problem~(\ref{eq_problem_equivalent}), the Lagrangian multipliers need to be adjusted. This is actually a dual problem with respect to Lagrangian multipliers, which can be optimized by the gradient method~\cite{bertsekas1997nonlinear}. For the $T$-th outer-loop, we have gotten the converged solution to  Lagrangian problem~(\ref{eq_problem_relaxed}) for given $\{ \lambda_i^{(T)} \}$ by inner-loop iterations, denoted by $\bar{\mathbf{x}}^{(T)}$, $\{\bar{l}_{ij}^{(T)}\}$, $\bar{P}_{\rm{B}}^{(T)}$ and $\{\bar{P}_{{i}}^{(T)}\}$. Then the multiplier $\lambda_i$ is updated by using the following formula~\cite{Fisher1981The}:
\begin{equation}\label{eq_update_lambda}
	\begin{aligned}
		\lambda_i^{(T+1)} = \max \left( 0, \lambda_i^{(T)} + \gamma^{(T)} \sum_{j \in \mathcal{J}_i} \bar{l}_{ij}^{(T)} (1-\bar{l}_{ij}^{(T)}) \right), 
	\end{aligned}
\end{equation} 
where $\gamma^{(T)}$ is the stepsize
formulated as 

{\color{black}
\begin{equation}\label{eq_update_stepsize}
	\begin{aligned}
		\gamma^{(T)} = \frac{\mu^{(T)} \times (q_{\rm{U}}^{(T)}-q_{\rm{L}}^{(T)})}{\sum_{i \in \mathcal{I}} \left( \sum_{j \in \mathcal{J}_i} \bar{l}_{ij}^{(T)} (1-\bar{l}_{ij}^{(T)}) \right)^2 },
	\end{aligned}
\end{equation}}
where $q_{\rm{U}}^{(T)}$ is the objective value of  Lagrangian problem~(\ref{eq_problem_relaxed}) at the $T$-th outer-loop, $q_{\rm{L}}^{(T)}$ is the objective value of  original problem~(\ref{eq_problem_equivalent}). The denominator is the sum square of the gradients. {\color{black}$\mu^{(T)}$ is an adaption parameter which is determined by setting $\mu^{(0)}=2$ and halving $\mu^{(T)}$ whenever $q_{\rm{U}}^{(T)}$ fails to decrease in the previous iteration.} If the solution to Lagrangian problem~(\ref{eq_problem_relaxed}) is not feasible for original problem~(\ref{eq_problem_equivalent}), which means the BS-UAV link or UAV-UE link is blocked, $q_{\rm{L}}^{(T)}$ is estimated to be zero. Otherwise, $q_{\rm{L}}^{(T)}$ is set to be the corresponding communication capacity by substituting the solution into original problem~(\ref{eq_problem_equivalent}).

{\color{black}
	It can be seen that (\ref{eq_update_lambda}) results an increase update direction for $\lambda_i$. In general, large $\lambda_i$ leads to small violation of constraint (\ref{c:lij2})~\cite{bertsekas1997nonlinear}. However, it is not effective to initially set $\lambda_i$ to be a vary large value, since the objective will be dominated by the penalty term $-\sum_{i \in \mathcal{I}} \lambda_i \sum_{j \in \mathcal{J}_i} l_{ij}(1-l_{ij})$ and the minimum communication capacity $R$ will be diminished. In contrast, initializing $\lambda_i$ with a small value will provide enough degree of freedom for UAV positioning  to obtain a good starting point. Then, by gradually increasing the value of $\lambda_i$ via (\ref{eq_update_lambda}), the upper bound given by Lagrangian problem~(\ref{eq_problem_relaxed}) can be gradually strengthened. Ultimately, constraint~(\ref{c:lij2}) is satisfied, which means that the UAV is positioned in a region without any building blockage. In this way, the solution to Lagrangian problem~(\ref{eq_problem_relaxed}) is also feasible for original problem~(\ref{eq_problem_equivalent}), and the gap between the two problem is decreased to zero.
	
}


\section{Overall Solution}\label{sec_overall}
In this section, we first provide the overall solution of the joint positioning and power allocation problem for UAV relay systems aided by geographic information. Then, the convergence and computational complexity are analyzed.
\subsection{Overall Solution}
The overall solution of the algorithm is summarized in Algorithm~\ref{alg:overall_solution}. In Line 1, we invoke Algorithm~\ref{alg:blckage_formula} to calculate the blocked regions for the BS and $K$ UEs caused by $M$ buildings based on geographic information. Then, in Lines 3-14, we employ the Lagrangian relaxation framework and update the Lagrangian multipliers $\{\lambda_i\}$ in an iterative way. Lines 5-10 solve  Lagrangian problem~(\ref{eq_problem_relaxed}) with given $\{\lambda_i^{(T)}\}$, where the power allocation and UAV positioning are optimized in an alternate manner. The inner-loop terminates, when the increase of the objective value of problem~(\ref{eq_problem_positioning_approx}) from one iteration to the next falls bellow a threshold $\epsilon_\mathrm{t}$ or the number of iterations exceeds the maximum value $L_\mathrm{t,max}$. The outer-loop terminates, when the difference between the objective value of problem~(\ref{eq_problem_relaxed}) $q_{\rm{U}}^{(T)}$ and the objective value of problem~(\ref{eq_problem_equivalent}) $q_{\rm{L}}^{(T)}$ falls bellow a threshold $\epsilon_\mathrm{T}$ or the number of iterations exceeds the maximum value $L_\mathrm{T,max}$. 

\begin{algorithm}[t] \small
	\label{alg:overall_solution}
	\caption{Joint positioning and power allocation problem for UAV relay systems.}
	\begin{algorithmic}[1]
		\REQUIRE ~The coordinates of the vertices of the buildings, $\{\mathbf{x}_k\}$, $\mathbf{x}_\mathrm{B}$, $P_\mathrm{B}^\mathrm{tot}$, $P_\mathrm{V}^\mathrm{tot}$, $x_\mathrm{D}$, $y_\mathrm{D}$, $h_\mathrm{min}$, $N_0$, $\alpha_1$, $\beta_1$, $\kappa_1$, $\epsilon_\mathrm{T}$, $\epsilon_\mathrm{t}$, $L_\mathrm{T,max}$, $L_\mathrm{t,max}$.
		\ENSURE ~$\mathbf{x}^\star$, $P_\mathrm{B}^\star$, $\{P_k^\star\}$. \\
		\STATE Calculate the blocked regions according to Algorithm~\ref{alg:blckage_formula}.
		\STATE Set the iteration index of outer-loops as $T=0$, and initialize $\{\lambda_i^{(0)}\}$, $\bar{\mathbf{x}}^{(0)}$, $\{\bar{l}_{ij}^{(0)}\}$.		
		\REPEAT
		\STATE Set the iteration index of inner-loops as $t=0$, and initialize $\rho^{(0)}$, $\mathbf{x}^{(0)}\leftarrow \bar{\mathbf{x}}^{(T)}$, $\{l_{ij}^{(0)}\}\leftarrow \{\bar{l}_{ij}^{(T)}\}$.
		\REPEAT
		\STATE Obtain $P_\mathrm{B}^{(t+1)}$ and $\{P_k^{(t+1)}\}$ according to Theorem~\ref{Theo_power}.
		\STATE Solve Problem~(\ref{eq_problem_positioning_approx}) and obtain $\mathbf{x}^{(t+1)}$, $\{l_{ij}^{(t+1)}\}$.
		\STATE Update $\rho^{(t+1)}\leftarrow \kappa_1 \rho^{(t)}$.
		\STATE Update $t\leftarrow t+1$.
		\UNTIL{The increase of objective value falls bellow $\epsilon_\mathrm{t}$ or $t>L_\mathrm{t,max}$.}	
		\STATE{Update $\bar{\mathbf{x}}^{(T)} \leftarrow \mathbf{x}^{(t)}$, $\{\bar{l}_{ij}^{(T)}\}\leftarrow \{l_{ij}^{(t)}\}$ }
		\STATE Update Lagrangian multiplier $\{\lambda_i^{(T+1)}\}$ based on (\ref{eq_update_lambda}).
		\STATE Update $T\leftarrow T+1$.
		\UNTIL{The difference between $q_{\rm{U}}^{(T)}$ and $q_{\rm{L}}^{(T)}$ falls below $\epsilon_\mathrm{T}$ or $T>L_\mathrm{T,max}$.}		
		\RETURN $R^\star$, $\mathbf{x}^\star$, $P_\mathrm{B}^\star$, $\{P_k^\star\}$, $\{\l_{ij}^\star\}$.
	\end{algorithmic}
\end{algorithm}

\subsection{Convergence Analysis}\label{sec_converge}

\subsubsection{Inner-loop Convergence}
As discussed in Section~\ref{sec_solution_relaxed},  Lagrangian problem~(\ref{eq_problem_relaxed}) is solved iteratively by applying BCD and SCA techniques. We first show that the objective function value of Lagrangian problem~(\ref{eq_problem_relaxed}) converges to a finite value. We use  $q_\mathrm{pos}^\mathrm{lb}$ to represent the objective values of approximate positioning problem~(\ref{eq_problem_positioning_approx}). Denote $\mathbf{P}=\{ P_\mathrm{B}, \{P_k,\forall k\in\mathcal{K} \} \}$, $\mathbf{Q}=\{ \mathbf{x}, \{l_{ij},\forall i\in\mathcal{I},\forall j\in\mathcal{J}_i \} \}$. On one hand, in Line 6 of Algorithm~\ref{alg:overall_solution}, since the closed-form optimal solution of (\ref{eq_problem_PowerControl}) is obtained for given $\mathbf{Q}^{(t)}$, we have
\begin{equation}\label{converge1}
q_{\rm{U}}^{}(\mathbf{P}^{(t)}, \mathbf{Q}^{(t)}) \leq q_{\rm{U}}^{}(\mathbf{P}^{(t+1)}, \mathbf{Q}^{(t)}).
\end{equation}
On the other hand, for given $\mathbf{P}^{(t+1)}$ in Line 7 of Algorithm~\ref{alg:overall_solution}, we have
\begin{equation}\label{converge2}
\begin{aligned}
q_{\rm{U}}^{}(\mathbf{P}^{(t+1)}, \mathbf{Q}^{(t)}) 
&\overset{\text{(a)}}{=} q_\mathrm{pos}^\mathrm{lb}(\mathbf{P}^{(t+1)}, \mathbf{Q}^{(t)})\\
&\overset{\text{(b)}}{\leq} q_\mathrm{pos}^\mathrm{lb}(\mathbf{P}^{(t+1)}, \mathbf{Q}^{(t+1)})\\
&\overset{\text{(c)}}{\leq} q_{\rm{U}}^{}(\mathbf{P}^{(t+1)}, \mathbf{Q}^{(t+1)}),
\end{aligned}
\end{equation}
where (a) holds because the first-order Taylor expansions in (\ref{eq_l_lb}), (\ref{eq_rate_lb}) and (\ref{eq_rate_BS_lb}) are tight at $\mathbf{Q}^{(t)}$. (b) holds because $\mathbf{Q}^{(t+1)}$ is the optimal solution of problem~(\ref{eq_problem_positioning_approx}) for given $\mathbf{P}^{(t+1)}$ and $\mathbf{Q}^{(t)}$.
(c) holds because the objective value of problem~(\ref{eq_problem_positioning_approx}) is the lower bound on that of its original problem~(\ref{eq_problem_positioning}) at $\mathbf{Q}^{(t+1)}$, under small neighborhood constraint~(\ref{trustregion}). Based on (\ref{converge1}) and (\ref{converge2}), we have
\begin{equation}
	\begin{aligned}
		q_{\rm{U}}^{}(\mathbf{P}^{(t)},\mathbf{Q}^{(t)}) \leq q_{\rm{U}}^{}(\mathbf{P}^{(t+1)}, \mathbf{Q}^{(t+1)}),
	\end{aligned}
\end{equation}
which indicates that the objective value of problem~(\ref{eq_problem_relaxed}) is non-decreasing during the iteration. Since the objective value is upper bounded by a finite value, it finally converges to a finite value.

Besides, under small neighborhood constraint~(\ref{trustregion}), the update of $\mathbf{x}^{(t+1)}$ is restricted in a small neighborhood around the local point $\mathbf{x}^{(t)}$ with a decreasing radius $\rho^{(t)}$. Therefore, we have
\begin{equation}
	\lim_{t\rightarrow \infty} \|\mathbf{x}^{(t+1)}-\mathbf{x}^{(t)}\| \leq \lim_{t\rightarrow \infty} \rho^{(t)}=0.
\end{equation}
Consequently, the solution of Lagrangian problem~(\ref{eq_problem_relaxed}) converges to a point.

\subsubsection{Outer-loop Convergence}
{\color{black}
The outer-loop is designed to solve a dual problem with respect to the Lagrangian multipliers by using gradient method. The convergence depends on the optimality of the Lagrangian problem~(\ref{eq_problem_relaxed}) solved by the inner-loop iterations.  
If Lagrangian problem~(\ref{eq_problem_relaxed}) can obtain globally optimal solution for given Lagrangian multipliers $\{\lambda_i\}$, then the gradient method by using the updating formula (\ref{eq_update_lambda}) and the step size (\ref{eq_update_stepsize}) can guarantee the convergence of the dual problem~\cite{held1974valida,Fisher1981The}. However, Lagrangian problem~(\ref{eq_problem_relaxed}) is non-convex and may not always converge to the globally optimal solution. 
As a result, gradient directions may not be optimal and the algorithm may not converge~\cite{zhao1997thesur}. 
Under this circumstance, 
initial condition plays an important role. If the algorithm fails to converge under the default initial condition, it will terminate when the number of outer-loop iterations exceeds a maximum value $L_\mathrm{T,max}$. Then the algorithm can be re-performed with another initialization. According to the simulation results, initially setting a larger altitude facilitates the convergence of the algorithm for non-convergence cases. Specially, if we initialize the algorithm with an unblocked position (i.e., constraints (\ref{c:block1})-(\ref{c:block4}) are satisfied), fix the values of binary variables $\{l_{ij}\}$, and only optimize  $\mathbf{x}$, $P_\mathrm{B}$, and $\{P_k\}$ via Algorithm~\ref{alg:overall_solution}, then the proposed scheme is guaranteed to converge to a locally feasible solution for original problem~(\ref{eq_problem_equivalent}). Since the penalty term $\sum_{i \in \mathcal{I}} \lambda_i \sum_{j \in \mathcal{J}_i} l_{ij}(1-l_{ij})=0$ holds for any given $\{\lambda_i\}$, the algorithm will terminate after one outer-loop iteration. From a physical point of view, the fixed values of $\{l_{ij}\}$ confine that the UAV has to be deployed into a locally small unblocked region.
The convergence of the proposed solution will be evaluated via simulations in Section~\ref{sec_simulation}.  

}



\subsection{Computational Complexity}
In the proposed joint UAV positioning and power allocation algorithm, Line 1 in Algorithm~\ref{alg:overall_solution} needs to calculate $(M+KM)$ blocked regions, which entails a computational complexity of $\mathcal{O}\left(KM \right)$. The main computational complexity of Lines 5-10 comes from solving positioning sub-problem~(\ref{eq_problem_positioning_approx}), whose complexity is $\mathcal{O}\left( (KM)^{3.5} \right)$ by using the interior point method~\cite{boyd2004convex}.
Line 12 needs to estimate whether the solution of Lagrangian problem~(\ref{eq_problem_relaxed}) is feasible for the original problem~(\ref{eq_problem}), i.e., whether the UAV's position is out of all the $(M+KM)$ blocked regions, which entails a computational complexity of $\mathcal{O}\left( KM \right)$. As a result, the overall computational complexity of Algorithm~\ref{alg:overall_solution} is 
$\mathcal{O}\left(L_\mathrm{T,max} L_\mathrm{t,max} \left(KM \right)^{3.5} \right).$

{\color{black}
As we can see, the overall computational complexity comes from involving auxiliary variables $\{l_{ij}\}$ to tackle the $(M+KM)$ blocked regions, i.e., constraint~(\ref{c:block}). In practice, due to the hardware and power constraints, the coverage area of a single UAV relay is usually small. Thus, the number of blocked regions may not be very large. Even for a typical area of interest that may consist of a lager number of structures, the proposed solution can also be used by dividing the target area into multiple sub-regions and employing multiple UAV relays for coverage. With this implementation, the number of blocked regions is reduced. On the other hand, if a single UAV has to be utilized to cover a large area with many structures (although this rarely happens),
we may use the following way to decrease the computational complexity. 
Constraint~(\ref{c:block}) can be rewritten in equivalent form as $\mathbf{x} \in \mathcal{D} \setminus \left(\bigcup_{i \in \mathcal{I}} \mathcal{D}_i\right)$.
Note that if there exist a $j\in \mathcal{I}$ and a $k\in \mathcal{I}$ ($j\neq k$) such that $\mathcal{D}_j \subseteq \mathcal{D}_k$, then we have $\bigcup_{i \in \mathcal{I}} \mathcal{D}_i=\bigcup_{i \in \mathcal{I},i\neq j} \mathcal{D}_i$, which means that the $j$-th blocked region $\mathcal{D}_j$ is \emph{redundant} and can be removed from $\bigcup_{i \in \mathcal{I}} \mathcal{D}_i$. Therefore, we can first search all the \emph{redundant} sets through pairwise comparison and remove them from $\bigcup_{i \in \mathcal{I}} \mathcal{D}_i$. 
From a physical point of view, a UE is usually blocked by nearby buildings and the blockage of distant ones does not make sense.
In this manner, the number of \emph{effective} blocked regions may be smaller than $M+KM$, and the number of auxiliary variables $\{l_{ij}\}$ to be introduced is also decreased. Thus, the overall computational complexity can be reduced.
}

\section{Performance Evaluation}\label{sec_simulation}
In this section, we provide simulation results to evaluate the performance of the proposed joint UAV positioning and power allocation scheme for UAV relay systems aided by geographic information.

\subsection{Simulation Setup and Benchmark Schemes}\label{subsec_setup}
We consider a Manhattan-like dense urban area with size $500\times500\text{ m}^2$, i.e., $x_\mathrm{D}=500$, $y_\mathrm{D}=500$. The center of each building follows $(100K_x-50,100K_y-50)$, with $K_x=1,...,5,K_y=1,...,5$. The length and width of each building are random variables which follow a uniform distribution in a range which is properly determined to reach the desired $20\%$ building density, i.e., \emph{the ratio of built-up land area to the total land area}~\cite{Al-Hourani2014optima}. The heights of the buildings are random variables that follow a Rayleigh distribution with an average height of 23 meters, where the distribution function is interceptive between 3 meters to 50 meters. Thus, the minimum flight altitude of the UAV relay $h_\mathrm{min}$ is set to 50 meters. The coordinates of the BS antenna are set as $\mathbf{x}_\mathrm{B}=(0,0,25)$. 
In (\ref{c:block1}), a sufficiently large constant $C$  is set to 
\begin{equation*}
	\begin{aligned}
	C=\max\limits_{\forall i\in \mathcal{I}, \forall j\in \mathcal{J}_i}~ 5 \left(b_{ij}-\mathbf{a}_{ij}^T \mathbf{x}\right),~ \mbox{s.t.}~ \mathbf{x}\in \mathcal{D}.
	\end{aligned}
\end{equation*}
Auxiliary variable $l_{ij}$ is initialized as $l_{ij}^{(0)}=(|\mathcal{J}_i|-1)/|\mathcal{J}_i|,  \forall j\in \mathcal{J}_i, \forall i\in \mathcal{I}$. The position of UAV relay is initialized as $\mathbf{x}^{(0)}=(\frac{x_\mathrm{D}}{2},\frac{y_\mathrm{D}}{2},500)$. In this way, the initial values are a feasible solution to problem~(\ref{eq_problem_positioning_approx}). Other adopted simulation parameter settings are summarized in Table~\ref{tab:para}~\cite{Al-Hourani2014optima, Al-Hourani2014modeli}, unless specified otherwise. The UEs are randomly generated in the areas without buildings, and each point in the simulation figures is the average performance over 500 UE distributions. 
{\color{black}
If a simulation fails to converge under the default initial condition, it will be re-performed under the following new setup. First, an initial position where the UAV is outside the blocked regions is obtained by searching the lowest altitude with setting $(x_\mathrm{V},y_\mathrm{V})=(\frac{x_\mathrm{D}}{2},\frac{y_\mathrm{D}}{2})$. Then, the obtained values of binary variables $\{l_{ij}\}$ are fixed during optimization, and only $\mathbf{x}$, $P_\mathrm{B}$, and $\{P_k\}$ are optimized by adopting Algorithm~\ref{alg:overall_solution}.
}

\begin{table}[h]
	\caption{Simulation Parameters}\label{tab:para}
	\footnotesize
	\begin{center}
		\begin{tabular}{|c|l|c|}
			\hline
			\textbf{Parameter}                        &\multicolumn{1}{c|}{\textbf{Description}}                                       & \textbf{Value} \\
			\hline
			$P_{\mathrm{B}}^{\mathrm{tot}}$                                 & Maximum transmit power at the BS                                   & 30 dBm \\
\hline
			$P_{\mathrm{V}}^{\mathrm{tot}}$                                 & Maximum transmit power at the UAV                                   & 30 dBm \\
\hline
			$N_0$                           & Power spectral density of the noise                             & -174 dBm/Hz \\
			\hline
			$f_c$                       & Carrier frequency                                         & 5 GHz \\
			\hline
			$W_{\mathrm{U}}$                      & Channel bandwidth of each UAV-UE link                                      & 5 MHz \\
\hline
			$W_{\mathrm{B}}$                      & Channel bandwidth of the BS-UAV link                                        & $K\times 5$ MHz \\
\hline
			$\alpha_{\mathrm{1}}$                    & Channel gain exponent for LoS path                          & 2 \\
			\hline
			$\beta_{\mathrm{1}}$                  & \begin{tabular}[c]{@{}l@{}}Channel gain at the reference distance \\of 1 m for LoS path\end{tabular}                          & -46.43 dB \\
			\hline
			$\lambda_i^{(0)}$                                   & \begin{tabular}[l]{@{}l@{}}Initial value of the Lagrangian multiplier \end{tabular}                              & 1 \\
\hline
			$\rho^{(0)}$                                   & \begin{tabular}[l]{@{}l@{}}Initial radius of the spherical region in\\ constraint~(\ref{trustregion}) \end{tabular}                              & 50 \\
\hline
			$\kappa_1$                                   & \begin{tabular}[l]{@{}l@{}}Step size for the radius reduction in\\ constraint~(\ref{trustregion}) \end{tabular}                               & 0.9  \\
			\hline
			$\epsilon_{\mathrm{t}}^{}$                                 & \begin{tabular}[l]{@{}l@{}}Threshold for convergence of inner-loop \end{tabular}                  & 0.01 \\
			\hline
			$\epsilon_{\mathrm{T}}^{}$                                 & \begin{tabular}[l]{@{}l@{}}Threshold for convergence of outer-loop \end{tabular}                  & 0.01 \\
			\hline
			$L_\mathrm{t,max}$                               & \begin{tabular}[l]{@{}l@{}}Maximum iteration number for inner-loop \end{tabular}                  & 30\\
\hline
$L_\mathrm{T,max}$                                 & \begin{tabular}[l]{@{}l@{}}Maximum iteration number for outer-loop \end{tabular}                  & 10\\
\hline
		\end{tabular}
	\end{center}
\end{table}
The proposed method is labeled by ``LR". {\color{black}Besides, four benchmark schemes are defined for performance comparison, namely ``3D-ES", ``2D-ES'', ``CENTER" and ``FREE", respectively. }
\begin{itemize}
	\item {\color{black}\emph{3D-ES}}: This scheme performs an exhaustive search over a 3-D lattice with $5$ meter spacing. In each lattice, the optimal transmit powers given in Theorem~\ref{Theo_power} are adopted at the BS and the UAV relay. The candidate coordinates which achieve the maximum communication capacity for the actual channel environment are chosen as the optimal solution for UAV positioning. {\color{black}This scheme serves as a performance upper bound on the minimum communication capacity for the UAV relay system.}

	\item {\color{black}\emph{2D-ES}: This scheme performs an exhaustive search over a 2-D horizontal grid with 5 meter spacing at a fixed altitude $H$. While other settings are the same as these of ``3D-ES'' scheme. This scheme is to show the necessity of altitude optimization.
	}
	
	\item {\color{black}\emph{CENTER}: This scheme sets the horizontal position of the UAV relay to be the center of the area, i.e., $(x_\mathrm{V},y_\mathrm{V})=(\frac{x_\mathrm{D}}{2},\frac{y_\mathrm{D}}{2})$ and performs an exhaustive search along the vertical plane to find the lowest altitude where the UAV relay is outside blocked regions. Then the optimal transmit powers given in Theorem~\ref{Theo_power} are adopted. This scheme is to show the necessity of horizontal position optimization.}
	
	\item {\color{black}\emph{FREE}: 
		This scheme assumes that  geographic information is not available for UAV positioning, and a joint 2-D positioning and power allocation is performed with a fixed altitude $H$. In other worlds, the blockage effect of buildings, i.e., constraint~(\ref{c:block}), is not considered in the problem. The problem is solved by applying BCD and SCA techniques, i.e., a similar manner shown in Lines 4-10 in Algorithm~\ref{alg:overall_solution}. The communication capacity is then calculated according to the actual LoS/NLoS channel environment at the obtained UAV position. This scheme is to verify the significance of geographic information for UAV positioning.}	
\end{itemize}

{\color{black}The proposed scheme can ensure that the obtained UAV position is out of all the blocked regions. However, the position obtained by other benchmark schemes may be blocked by buildings. NLoS channels are required for such circumstance. Therefore, for performance comparison, the NLoS channel gains of BS-UAV and UAV-UE links are respectively given by
\begin{equation}\label{eq_channel_BS_NLoS}
	g_{\mathrm{B}}^{}=
	\beta_2 \|\mathbf{x}-\mathbf{x}_\mathrm{B}\|^{-\alpha_2}, 
\end{equation}
\begin{equation}\label{eq_channel_User_NLoS}
	g_k=
	\beta_2 \|\mathbf{x}-\mathbf{x}_k\|^{-\alpha_2}, 
\end{equation}
where $\alpha_2$ is the path loss exponent, and $\beta_2$ is the channel gain at reference distance of $1$ m. We set the parameters as $\alpha_{\mathrm{2}}=3.3$, and $\beta_{\mathrm{2}}=-56.43$ dB. We adopt the LoS channel models
in (\ref{eq_channel_BS}) and (\ref{eq_channel_User}), where an LoS path exists for the BS-UAV and UAV-UE links if the UAV is located out of the blocked regions. Otherwise, NLoS channels in (\ref{eq_channel_BS_NLoS}) and (\ref{eq_channel_User_NLoS}) are employed.}



\subsection{Simulation Results}\label{subsec_result}
First, we provide a demonstration for the optimization process of the UAV positioning in Fig.~\ref{fig:demo}, where 16 UEs are randomly distributed in the areas without buildings. The star labeled \emph{``Start"} is the initial position of the UAV, and the star labeled \emph{``End"} is the final position. Other stars represent the positions during the iteration. Different line types are used to distinguish the steps of outer-loops. It can be observed that the solution converges within three outer iterations. At the beginning of the iteration, with small initial Lagrangian multipliers, the UAV relay moves toward the ideal position without considering the existence of building blockages. As we can see, the UAV tends to decrease its height so as to get close to the ideal position. As the first outer-loop ends, the violation of blockage constraints leads to larger Lagrangian multipliers. Then, in the following outer-loop, the blockage constraints make sense and the UAV has to increase its height as well as adjust its horizontal position to avoid blockages. Finally, the position of the UAV relay is located outside blocked regions.
\begin{figure}[t]
	\begin{center}
		\includegraphics[width=\figwidth cm]{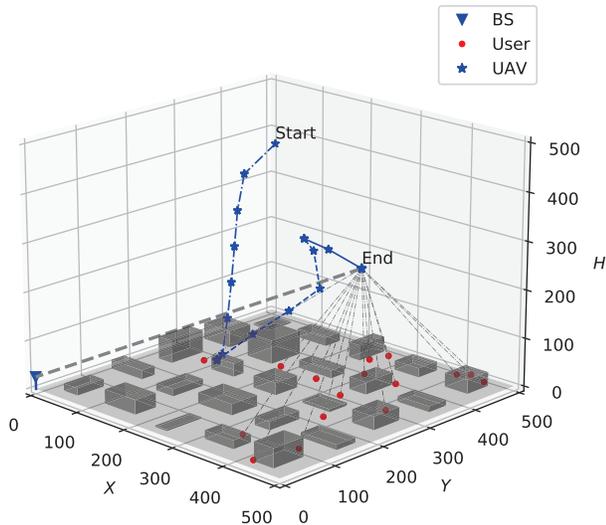}
		\caption{{\color{black}Demonstration for UAV positioning with 16 UEs.}}
		\label{fig:demo}
	\end{center}
	\vspace*{10pt}
\end{figure}


Fig.~\ref{fig:converge_out} demonstrates the convergence of the proposed Algorithm~\ref{alg:overall_solution} versus the number of outer-loop iterations  for different numbers of UEs. {\color{black}Simulation results show that the algorithm converges well for the vast majority of cases under the default initialization given in Section~\ref{subsec_setup} ($100\%$ for 1 UE, $98.8\%$ for 8 UEs, and $93.6\%$ for 32 UEs).}
As can be observed, the duality gap between the objective value of Lagrangian problem~(\ref{eq_problem_relaxed}) $q_\mathrm{U}^{}$ and the objective value of original problem~(\ref{eq_problem_equivalent}) $q_\mathrm{B}^{}$ decreases with the iteration and converges within 5 iterations for all settings. The convergence of the outer-loop iteration means that appropriate values of the Lagrangian multipliers are obtained to ensure the UAV relay out of the blocked regions.
\begin{figure}[t]
	\begin{center}
		\includegraphics[width=\figwidth cm]{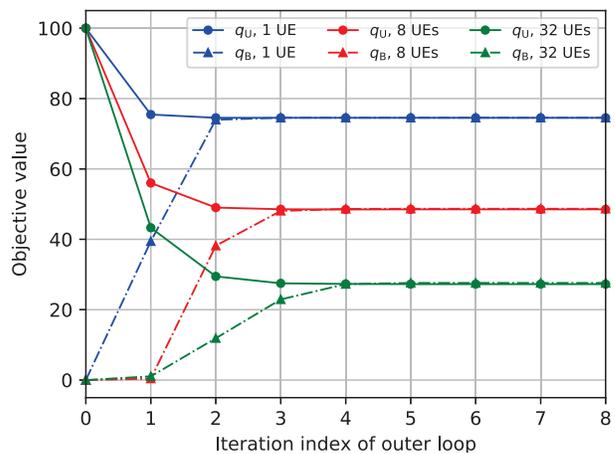}
		\caption{{\color{black}Evaluation of the convergence of the outer-loop iteration of the proposed Algorithm~\ref{alg:overall_solution} for different numbers of UEs.}}
		\label{fig:converge_out}
	\end{center}
\end{figure}

\begin{figure}[t]
	\begin{center}
		\includegraphics[width=\figwidth cm]{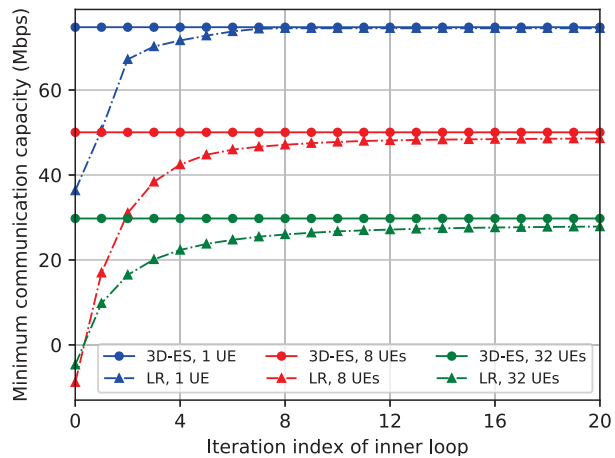}
		\caption{{\color{black}Evaluation of the convergence of the inner-loop iteration of the proposed Algorithm~\ref{alg:overall_solution} for different numbers of UEs.}}
		\label{fig:converge_in}
	\end{center}
\end{figure}
In Fig.~\ref{fig:converge_in}, we evaluate the convergence of the last inner-loop iteration of Algorithm~\ref{alg:overall_solution} for different numbers of UEs. For comparison, the results of ``3D-ES" scheme are taken as upper bounds. It can be observed that the proposed method converges to a value close to the upper bound within 20 iterations for all settings. As the number of UEs increases, the gap increases slightly (0.25 Mbps for 1 UE, 1.34 Mbps for 8 UEs, and 1.75 Mbps for 32 UEs). This is because the unblocked region becomes narrow and decentralized with the increasing of the number of UEs, leading to multiple locally optimal locations. Even though, our algorithm achieves more than $94\%$ capacity performance of the upper bound ($99.7\%$ for 1 UE, $97.3\%$ for 8 UEs, and $94.1\%$ for 32 UEs). The results confirm that with the optimized Lagrangian multipliers, the inner-loop for Lagrangian problem~(\ref{eq_problem_relaxed}) can not only avoid blocked regions, but also obtain a near-optimal solution for original problem~(\ref{eq_problem_equivalent}).

Fig.~\ref{fig:RvsUsers} compares the minimum communication capacities for different schemes versus the number of UEs. From the results, the performance of the proposed joint positioning and power allocation method  aided by geographic information is very close to the performance upper bound given by ``3D-ES" scheme, and outperforms all other benchmark schemes. In addition, as the number of UEs increases, the minimum communication capacity decreases. The reason is as follows. On one hand, the maximum transmit power at the BS $P_{\mathrm{B}}^{\mathrm{tot}}$ and at the UAV relay $P_{\mathrm{V}}^{\mathrm{tot}}$ are constant. As $K$ increases, the transmit power that can  be potentially allocated to each UE is reduced, and thus leads to lower capacity. On the other hand, as $K$ increases, more blocked regions are involved. The UAV relay tends to be deployed at a higher altitude to avoid blockages, and thus leads to higher path loss and lower capacity. {\color{black} 
The performance of ``2D-ES'' deteriorates dramatically as the number of UEs increases. The reason is as follows. On one hand, the fixed altitude $H$ may not be the optimal altitude when the number of UEs is small, and on the other hand, there may not exist any unblocked region when the number of UEs is large. 
For ``CENTER'' scheme, the UAV relay has to be deployed at a very high altitude to avoid blockages, and thus leads to higher path loss and lower capacity. 
Finally, without geographic information, blockage effect cannot be properly considered during optimization, and thus ``FREE'' scheme fails to guarantee practical communication performance and gets the worst results.}

\begin{figure}[t]
	\begin{center}
		\includegraphics[width=\figwidth cm]{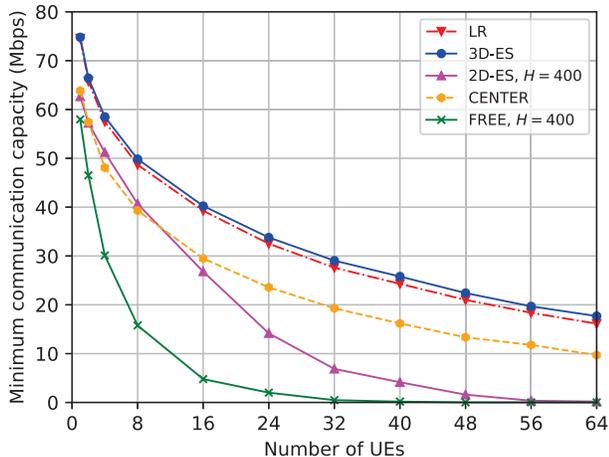}
		\caption{{\color{black}Minimum communication capacities for different schemes versus the number of UEs.}}
		\label{fig:RvsUsers}
	\end{center}
\end{figure}

Fig.~\ref{fig:RvsPbs} compares the minimum communication capacities for different schemes versus the maximum transmit powers at the BS ($P_{\mathrm{B}}^{\mathrm{tot}}$), where the number of UEs is 8. 
As can be observed again, the proposed method achieves a performance close to the upper bound and outperforms all other benchmark schemes. As $P_{\mathrm{B}}^{\mathrm{tot}}$ increases, the minimum communication capacity of the proposed method improves, but with a decreased rate of improvement. The reason is as follows. When $P_{\mathrm{B}}^{\mathrm{tot}}$ is small, the overall capacity is limited by the BS-UAV link. The UAV has to be positioned close to the BS. As $P_{\mathrm{B}}^{\mathrm{tot}}$ increases, the limitation of the BS-UAV link is lightened, and the UAV positioning has more freedom to get a better performance. When $P_{\mathrm{B}}^{\mathrm{tot}}$ is sufficiently large, the capacity of the BS-UAV link does not restrict the system performance, which, however, is limited by the UAV-UE links for fixed $P_{\mathrm{V}}^{\mathrm{tot}}$. Therefore, the capacity does not increase as $P_{\mathrm{B}}^{\mathrm{tot}}$ is larger than a certain threshold.

\begin{figure}[t]
	\begin{center}
		\includegraphics[width=\figwidth cm]{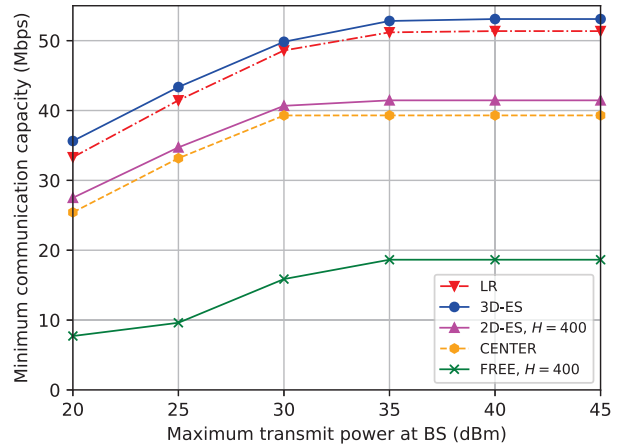}
		\caption{{\color{black}Minimum communication capacities for different schemes versus BS transmit powers.}}
		\label{fig:RvsPbs}
	\end{center}
\end{figure}

Fig.~\ref{fig:RvsPuav} compares the minimum communication capacities for different schemes versus the maximum transmit powers at the UAV ($P_{\mathrm{V}}^{\mathrm{tot}}$), where the number of UEs is 8. The proposed method still achieves a performance close to the upper bound and outperforms all the other benchmark schemes. When $P_{\mathrm{V}}^{\mathrm{tot}}$ is small, the communication capacity is limited by the UAV-UE links. As $P_{\mathrm{V}}^{\mathrm{tot}}$ increases, the capacity of UAV-UE links becomes larger, thus leading to a better performance of the system. Note that the ``FREE" scheme gets even worse performance as $P_{\mathrm{V}}^{\mathrm{tot}}$ increases from 30 dBm to 35 dBm. The reason is as follows. 
When $P_{\mathrm{V}}^{\mathrm{tot}}$ increases from 30 dBm to 35 dBm, the overall capacity is limited by the BS-UAV link. The UAV has to be positioned close to the BS, where UAV-UE links are more likely to be blocked by buildings. As $P_{\mathrm{V}}^{\mathrm{tot}}$ continues to increase, even though the UAV-UE link is still blocked, the signal-to-noise ratio (SNR) of the UE receiver through NLoS link increases, thus leading to better performance. In contrast, the proposed method takes the blockage effect into account during the optimization and always obtains satisfactory performance for difference transmit powers. The results show the significance of the geographic information for the UAV relay systems.

\begin{figure}[t]
	\begin{center}
		\includegraphics[width=\figwidth cm]{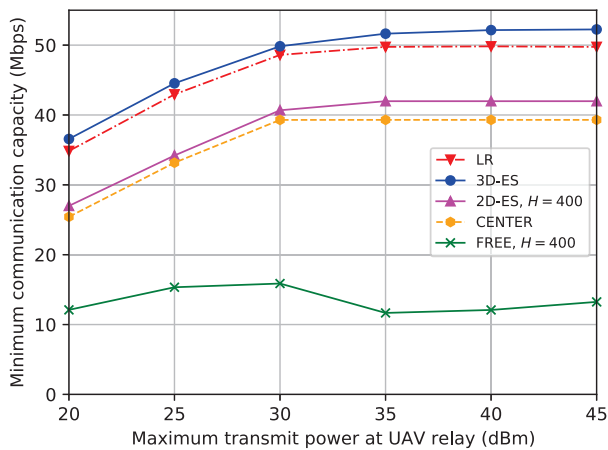}
		\caption{{\color{black}Minimum communication capacities for different schemes versus UAV transmit powers.}}
		\label{fig:RvsPuav}
	\end{center}
\end{figure}

\begin{figure}[t]
	\begin{center}
		\includegraphics[width=\figwidth cm]{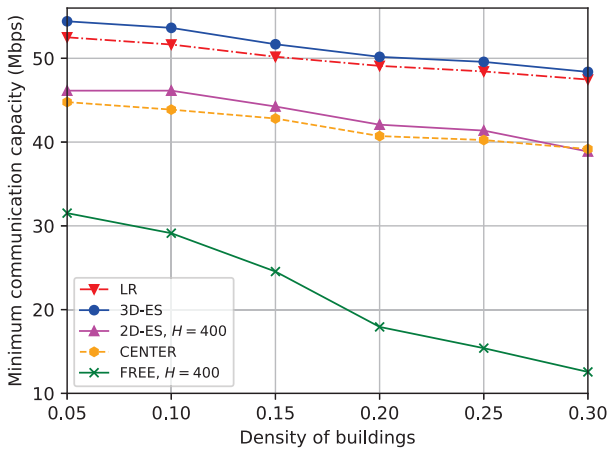}
		\caption{{\color{black}Minimum communication capacities for different schemes versus building density.}}
		\label{fig:RvsDensity}
	\end{center}
\end{figure}

Finally, we compare the minimum communication capacity for different schemes versus the density of buildings in Fig.~\ref{fig:RvsDensity}, where the number of UEs is 8. The length and width of each building are random variables which follow a uniform distribution in a range which is properly determined to reach the desired building density. As can be observed, the minimum communication capacities for the considered schemes all decrease as the density of building increases. This is because denser buildings lead to more severe blockage effect. The UAV relay has to adjust its horizontal position and increase its altitude to avoid blockages. Even under this scenario, the proposed method still obtains a performance  close to the upper bound and outperforms all other benchmark schemes. 
{\color{black}Besides,  the decreasing rate of the capacity obtained by the proposed method is similar to that of ``3D-ES'' scheme}, which suggests that the proposed method can efficiently get satisfying performance with the aid of geographic information for different building density circumstances.

\section{Conclusion}\label{sec_conclusion}

In this paper, we proposed to employ a UAV relay to improve the minimum communication capacity from a ground BS to multiple UEs. Assisted by geographic information, the blockage effect caused by buildings was modeled. A joint optimization problem was formulated for the UAV positioning and power allocation  under the constraints of link capacity, maximum transmit power, and blockage, to maximize the minimum communication capacity among all the UEs. To solve this non-convex problem, we proposed to leverage Lagrangian relaxation and performed a two-loop optimization framework. In the outer-loop, the Lagrange multipliers were adaptively updated to ensure the Lagrangian problem converge to the tightest upper bound on the original problem. In the inner-loop, the Lagrangian problem was solved by applying the BCD technique, where UAV positioning and power allocation were alternately solved. Simulation results demonstrated that the proposed joint positioning and power allocation method aided by geographic information can closely approach a performance upper bound and outperforms three benchmark schemes in terms of the minimum communication capacity. The proposed blockage modeling method and the optimization framework have general significance for UAV positioning aided by geographic information.

\appendices

{\color{black}
\section{Proof of Lemma \ref{Lemma_LR_upperbound}}\label{app_LR_upperbound}
By replacing binary constraint~(\ref{c:block2}) with constraints~(\ref{c:lij1}) and (\ref{c:lij2}), and subtracting the term $\sum_{i \in \mathcal{I}} \lambda_i \sum_{j \in \mathcal{J}_i} l_{ij}(1-l_{ij})$ to the objective function of original problem~(\ref{eq_problem_equivalent}), we obtain 
\begin{align}
	\max\limits_{\mathbf{x}, P_\mathrm{B}, \{P_k\}, \{l_{ij}\}}~~ & R-\sum_{i \in \mathcal{I}} \lambda_i \sum_{j \in \mathcal{J}_i} l_{ij}(1-l_{ij})   \label{eq_problem_middle} \\
	\mbox{s.t.}~~ &\lambda_i \geq 0, \forall i \in \mathcal{I},  \notag \\
	& \text{(\ref{c:rate1}), (\ref{c:rate2}), (\ref{c:p1}), (\ref{c:p2}), (\ref{c:p3}), }   \notag  \\
	&\text{(\ref{c:block1}), (\ref{c:block3}), (\ref{c:block4}), (\ref{c:lij1}), (\ref{c:lij2})}.   \notag 
\end{align}

It can be seen that problem~(\ref{eq_problem_middle}) has the same feasible region as that of original problem~(\ref{eq_problem_equivalent}) as constraints (\ref{c:lij1}) and (\ref{c:lij2}) are equivalent to (\ref{c:block2}). Note that constraints~(\ref{c:lij1}) and (\ref{c:lij2}) together confine the value of any in $\{l_{ij}\}$ to either $0$ or $1$. Then, $\{l_{ij}(1-l_{ij})=0\}$ holds for any feasible solution of problem~(\ref{eq_problem_middle}) (which is also feasible for original problem~(\ref{eq_problem_equivalent})). Therefore, the optimal value of problem~(\ref{eq_problem_middle}) for any given Lagrangian multipliers is the same as that of original problem~(\ref{eq_problem_equivalent}).

Further, by removing constraint~(\ref{c:lij2}) from problem~(\ref{eq_problem_middle}), we obtain Lagrangian problem~(\ref{eq_problem_relaxed}). Since removing constraint~(\ref{c:lij2}) expands the feasible region and does not change the objective function,
the optimal value of Lagrangian problem~(\ref{eq_problem_relaxed}) for any given Lagrangian multipliers is an upper bound on that of problem~(\ref{eq_problem_middle}), and is also an upper bound on that of original problem~(\ref{eq_problem_equivalent}).

This completes the proof.}

\section{Proof of Theorem \ref{Theo_power}}\label{app_power}
First, to maximize the minimum communication capacity $R$, at least one of the transmit powers at the BS and UAV relay must reach the maximum possible value. The reason is as follows. Assume that the optimal transmit powers at the BS and UAV relay are both smaller than their maximum values, i.e., 
$P_{\mathrm{B}}^{\star} < P_{\mathrm{B}}^{\mathrm{tot}}$ and $ \sum_{k\in \mathcal{K}} P_{k}^{\star} < P_{\mathrm{V}}^{\mathrm{tot}}$. Then, there exists a small positive value $\delta$ which ensures $(1+\delta)P_{\mathrm{B}}^{\star}$ and  $\sum_{k\in \mathcal{K}} (1+\delta) P_{k}^{\star}$ do not exceed the maximum values of their transmit powers. It can be verified that $\left( (1+\delta)P_{\mathrm{B}}^{\star}, \{(1+\delta) P_{k}^{\star}\} \right)$ yield a larger minimum communication capacity than $\left(P_{\mathrm{B}}^{\star}, \{P_{k}^{\star}\} \right)$, which contradicts the assumption that  $\left(P_{\mathrm{B}}^{\star}, \{P_{k}^{\star}\} \right)$ is optimal. 

Then, we show that with  $\sum_{k\in \mathcal{K}} P_{k}\triangleq  P_{\mathrm{V}}\leq P_{\mathrm{V}}^{\mathrm{tot}}$, the optimal power allocation $\{P_k^\star\}$ to maximize the minimum communication capacity among all the UAV-UE link satisfies
\begin{equation}\label{eq_UAV_power}
	\eta_1^{(t)} P_1^\star  = \eta_2^{(t)} P_2^\star= ... = \eta_K^{(t)} P_K^\star.	
\end{equation}
The reason is as follows. Assume that (\ref{eq_UAV_power}) is not the optimal solution. Then at least one equality relationship in (\ref{eq_UAV_power}) is not true. By sorting $\{P_k^\star\}$ in ascending order, they follow
\begin{equation}\label{eq_UAV_power_ineq}
	\eta_{\pi_1}^{(t)} P_{\pi_1}^\star =... = \eta_{\pi_i}^{(t)}P_{\pi_{i}}^\star < \eta_{\pi_{i+1}}^{(t)} P_{\pi_{i+1}}^\star \leq ... \leq \eta_{\pi_K}^{(t)} P_{\pi_K}^\star,
\end{equation}
where $\{\pi_k\}$ are the indices after the proper permutation, and the first strict inequality occurs after the $\pi_i$-th term. Then, there exists a positive value $\Delta$ which ensures that 
\begin{align}
	&\eta_{\pi_1}^{(t)} (P_{\pi_1}^\star+\Delta) =... = \eta_{\pi_i}^{(t)}(P_{\pi_{i}}^\star+\Delta)= \notag\\ &\eta_{\pi_{i+1}}^{(t)} (P_{\pi_{i+1}}^\star-i\Delta) \leq ... \leq \eta_{\pi_K}^{(t)} P_{\pi_K}^\star, 
\end{align}
and yields a larger minimum communication capacity than $\{P_{k}^{\star}\}$, which contradicts the assumption that $\{P_{k}^{\star}\}$ is the optimal solution. Thus, (\ref{eq_UAV_power}) always holds for the optimal solution. 

On one hand, the minimum communication capacity constrained by the BS-UAV link is 
\begin{align}
	\widetilde{R}=\frac{W_{\mathrm{B}}^{}}{K} \log_2\left(1+\eta_\mathrm{B}^{(t)} P_{\mathrm{B}}^{}\right).
\end{align}
On the other hand, combining $\sum_{k\in \mathcal{K}} P_{k} = P_{\mathrm{V}}^{}$ with (\ref{eq_UAV_power}), the minimum communication capacity constrained by the UAV-UE link is
\begin{align}
	\widehat{R}=W_{\mathrm{U}}^{} \log_2\left(1+\eta_\mathrm{V}^{(t)} P_{\mathrm{V}}^{} \right), 
\end{align}
and the corresponding power allocation is $\widehat{P}_k=\eta_\mathrm{V}^{(t)} P_{\mathrm{V}}^{}/\eta_k^{(t)}$. 

When $\left(1+\eta_\mathrm{B}^{(t)} P_{\mathrm{B}}^{\mathrm{tot}}\right)^{\frac{W_{\mathrm{B}}^{}}{K}}< \left(1+\eta_\mathrm{V}^{(t)} P_{\mathrm{V}}^{\mathrm{tot}} \right)^{W_{\mathrm{U}}^{}}$, we have $\widetilde{R}<\widehat{R}$. Thus, $P_{\mathrm{B}}^{\star}=P_{\mathrm{B}}^{\mathrm{tot}}$ maximizes the communication capacity of the BS-UAV link. Meanwhile, to avoid the waste of transmit power, $P_{\mathrm{V}}$ should be reduced to 
\begin{align}
P_{\mathrm{V}}^\star =\left( \left(1+\eta_\mathrm{B}^{(t)} P_{\mathrm{B}}^{\mathrm{tot}}\right)^{\frac{W_{\mathrm{B}}^{}}{KW_{\mathrm{U}}}}-1 \right)/\eta_\mathrm{V}^{(t)}.
\end{align}
Thus, the optimal power allocation at the UAV relay is 
\begin{align}
P_{k}^{\star}=\left( \left(1+\eta_\mathrm{B}^{(t)} P_{\mathrm{B}}^{\mathrm{tot}}\right)^{\frac{W_{\mathrm{B}}^{}}{KW_{\mathrm{U}}}}-1 \right)/\eta_k^{(t)}.
\end{align}

When $\left(1+\eta_\mathrm{B}^{(t)} P_{\mathrm{B}}^{\mathrm{tot}}\right)^{\frac{W_{\mathrm{B}}^{}}{K}}\geq \left(1+\eta_\mathrm{V}^{(t)} P_{\mathrm{V}}^{\mathrm{tot}} \right)^{W_{\mathrm{U}}^{}}$, we have $\widetilde{R} \geq \widehat{R}$. Thus, $P_{\mathrm{V}}^{\star}=P_{\mathrm{V}}^{\mathrm{tot}}$ maximizes the communication capacity of the UAV-UE link. Meanwhile, to avoid the waste of transmit power, $P_{\mathrm{B}}$ should be reduced to 
\begin{align}
P_{\mathrm{B}}^\star = \left( \left(1+\eta_\mathrm{V}^{(t)} P_{\mathrm{V}}^{\mathrm{tot}} \right)^{\frac{KW_{\mathrm{U}}^{}}{W_{\mathrm{B}}^{}}} - 1 \right) / \eta_\mathrm{B}^{(t)}.
\end{align}
In addition, the optimal power allocation of the UAV relay is
\begin{align}
P_k^\star=\eta_\mathrm{V}^{(t)}P_{\mathrm{V}}^{\mathrm{tot}}/\eta_k^{(t)}.
\end{align}

This completes the proof.

\bibliographystyle{IEEEtran} 

\vspace{11pt}

\begin{IEEEbiography}[{\includegraphics[width=1in,height=1.25in,clip,keepaspectratio]{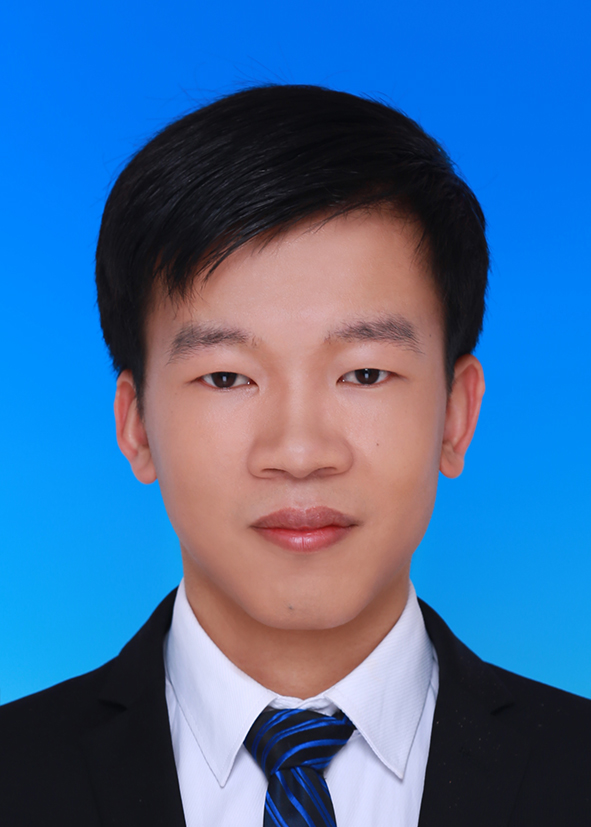}}]{Pengfei Yi}
	 received the B.E. degree from the Department of Electronics and Information, Northwestern
	Polytechnical University, Xian, China, in 2016, and the M.E. degree from the Department of Information and electronics, Beijing Institute of Technology, Beijing, China, in 2018. He is currently pursuing the Ph.D. degree with the Department of Electronics and Information Engineering, Beihang University, Beijing, China. His research interests include UAV communications and networking.
\end{IEEEbiography}
\vspace{11pt}
\begin{IEEEbiography}[{\includegraphics[width=1in,height=1.25in,clip,keepaspectratio]{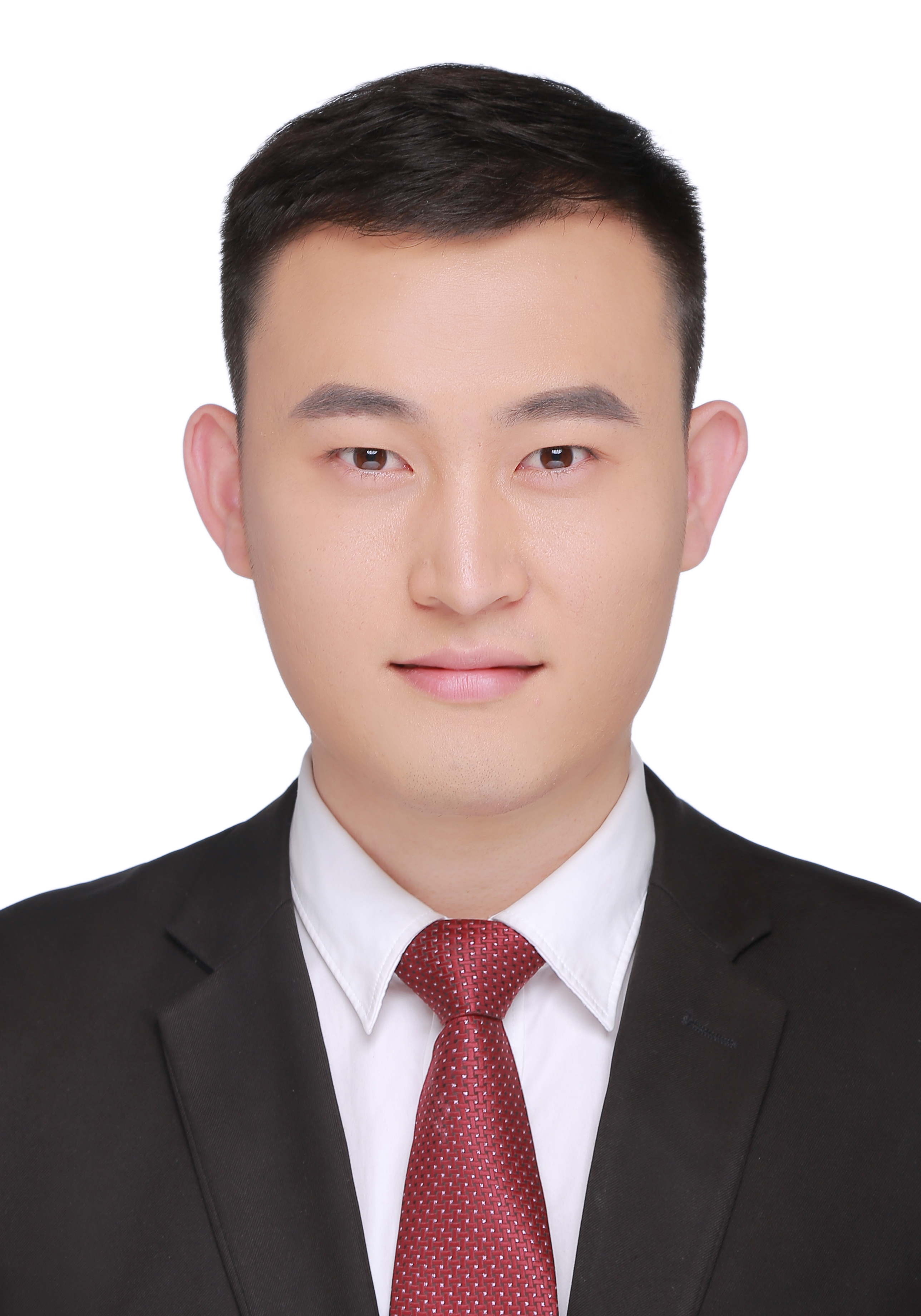}}]{Lipeng Zhu}
	(Member, IEEE) received the B.S. degree in the Department of Mathematics and System Sciences from Beihang University in 2017, and the Ph.D. degree in the Department of Electronic and Information Engineering from Beihang University in 2021. He is currently a Research Fellow with the Department of Electrical and Computer Engineering, National University of Singapore. His research interest is millimeter-wave communications, non-orthogonal multiple access, and UAV communications.
\end{IEEEbiography}
\vspace{11pt}
\begin{IEEEbiography}[{\includegraphics[width=1in,height=1.25in,clip,keepaspectratio]{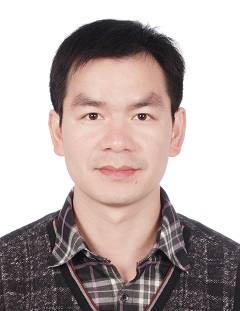}}]{Zhenyu Xiao}
	(M'11, SM'17) received the B.E. degree with the Department of Electronics and Information Engineering, Huazhong University of Science and Technology, Wuhan, China, in 2006, and the Ph.D. degree with the Department of Electronic Engineering, Tsinghua University, Beijing, China, in 2011. From 2011 to 2013, he held a postdoctorial position with the Department of Electronic Engineering, Tsinghua University. He was with the School of Electronic and Information Engineering, Beihang University, Beijing, as a Lecturer, from 2013 to 2016, and an Associate Professor from 2016 to 2020, where he is currently a Full Professor. He has visited the University of Delaware from 2012 to 2013 and the Imperial College London from 2015 to 2016.
	
	Dr. Xiao has authored or coauthored over 70 papers including IEEE JSAC, IEEE TWC, IEEE TSP, IEEE TVT, IEEE COMML, IEEE WCL, IET COMM etc. He has been a TPC member of IEEE GLOBECOM, IEEE WCSP, IEEE ICC, IEEE ICCC etc. He is currently an associate editor for IEEE Transactions on Cognitive Communications and Networking, China Communications, IET Communications, KSII transactions on Internet and Information Systems, Frontiers in Communications and Networks. He has also been the Lead Guest Editor for the Special Issue on Antenna Array Enabled Space/Air/Ground Communications and Networking of IEEE Journal on Selected Areas in Communications, one named Space-Air-Ground Integrated Network with Native Intelligence (NI-SAGIN): Concept, Architecture, Technology, and Radio of China Communications, and one named LEO Satellite Constellation Networks of Frontiers in Communications and Networks. Dr. Xiao has received 2017 Best Reviewer Award of IEEE TWC, 2019 Exemplary Reviewer Award of IEEE WCL, and the 4th China Publishing Government Award. He has received the Second Prize of National Technological Invention, the First Prize of Technical Invention of China Society of Aeronautics and Astronautics, and the Second Prize of Natural Science of China Electronics Society. Dr Xiao is an active researcher with broad interests on millimeter wave communications, UAV/satellite communications and networking, etc. He is an IEEE senior member, and was elected as one of 2020 highly cited Chinese researchers.
\end{IEEEbiography}
\vspace{11pt}
\begin{IEEEbiography}[{\includegraphics[width=1in,height=1.25in,clip,keepaspectratio]{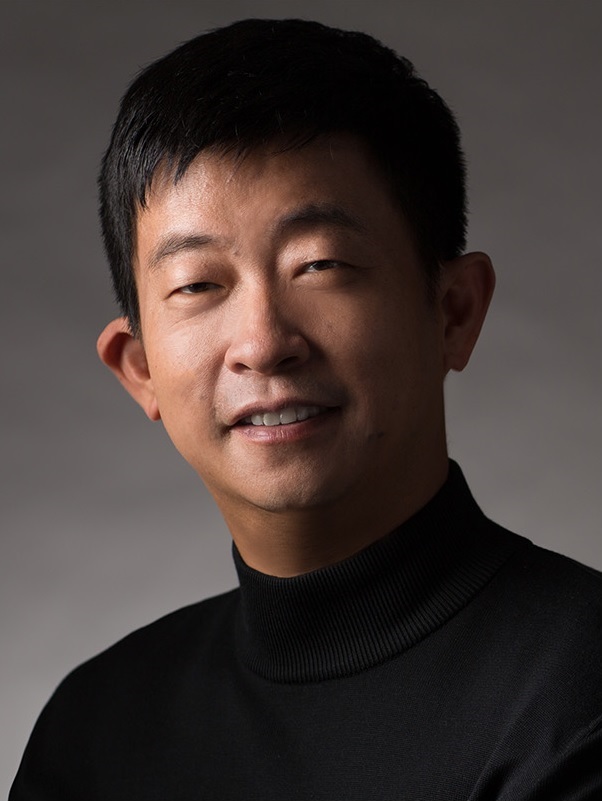}}]{Zhu Han}
	 (S’01–M’04-SM’09-F’14) received the B.S. degree in electronic engineering from Tsinghua University, in 1997, and the M.S. and Ph.D. degrees in electrical and computer engineering from the University of Maryland, College Park, in 1999 and 2003, respectively. 
	
	From 2000 to 2002, he was an R\&D Engineer of JDSU, Germantown, Maryland. From 2003 to 2006, he was a Research Associate at the University of Maryland. From 2006 to 2008, he was an assistant professor at Boise State University, Idaho. Currently, he is a John and Rebecca Moores Professor in the Electrical and Computer Engineering Department as well as in the Computer Science Department at the University of Houston, Texas. His research interests include wireless resource allocation and management, wireless communications and networking, game theory, big data analysis, security, and smart grid.  Dr. Han received an NSF Career Award in 2010, the Fred W. Ellersick Prize of the IEEE Communication Society in 2011, the EURASIP Best Paper Award for the Journal on Advances in Signal Processing in 2015, IEEE Leonard G. Abraham Prize in the field of Communications Systems (best paper award in IEEE JSAC) in 2016, and several best paper awards in IEEE conferences. Dr. Han was an IEEE Communications Society Distinguished Lecturer from 2015-2018, AAAS fellow since 2019, and ACM distinguished Member since 2019. Dr. Han is a 1\% highly cited researcher since 2017 according to Web of Science. Dr. Han is also the winner of the 2021 IEEE Kiyo Tomiyasu Award, for outstanding early to mid-career contributions to technologies holding the promise of innovative applications, with the following citation: ``for contributions to game theory and distributed management of autonomous communication networks."
\end{IEEEbiography}
\vspace{11pt}
\begin{IEEEbiography}[{\includegraphics[width=1in,height=1.25in,clip,keepaspectratio]{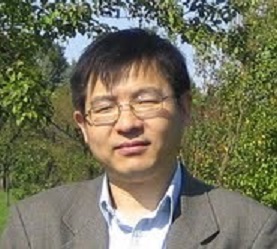}}]{Xiang-Gen Xia}
 (M'97, SM'00, F'09) received his B.S. degree in mathematics from Nanjing Normal University, Nanjing, China, and his M.S. degree in mathematics from Nankai University, Tianjin, China, and his Ph.D. degree in electrical engineering from the University of Southern California, Los Angeles, in 1983, 1986, and 1992, respectively.

He was a Senior/Research Staff Member at Hughes Research Laboratories, Malibu, California, during 1995-1996. In September 1996, he joined the Department of Electrical and Computer Engineering, University of Delaware, Newark, Delaware, where he is the Charles Black Evans Professor. His current research interests include space-time coding, MIMO and OFDM systems, digital signal processing, and SAR and ISAR imaging. Dr. Xia is the author of the book Modulated Coding for Intersymbol Interference Channels (New York, Marcel Dekker, 2000).

Dr. Xia received the National Science Foundation (NSF) Faculty Early Career Development (CAREER) Program Award in 1997, the Office of Naval Research (ONR) Young Investigator Award in 1998, and the Outstanding Overseas Young Investigator Award from the National Nature Science Foundation of China in 2001. He also received the Outstanding Junior Faculty Award of the Engineering School of the University of Delaware in 2001. He is currently serving and has served as an Associate Editor for numerous international journals including IEEE Wireless Communications Letters, IEEE Transactions on Signal Processing, IEEE Transactions on Wireless Communications, IEEE Transactions on Mobile Computing, and IEEE Transactions on Vehicular Technology. Dr. Xia is Technical Program Chair of the Signal Processing Symp., Globecom 2007 in Washington D.C. and the General Co-Chair of ICASSP 2005 in Philadelphia.

\end{IEEEbiography}
\end{document}